\documentclass[10pt]{article}
\usepackage[utf8]{inputenc}
\usepackage{subfiles}
\usepackage[english]{babel}
\usepackage[a4paper,top=3 cm,bottom=3 cm,inner=2.5 cm,outer=2.5 cm]{geometry} 
\usepackage{graphicx}
\usepackage{subfigure, tikz}
\usetikzlibrary{patterns}
\usetikzlibrary{positioning}
\usetikzlibrary{decorations.markings}
\usepackage{mathtools}
\usepackage{amsmath, amssymb, amsthm} 
\usepackage{amsfonts}
\usepackage{mathbbol}
\usepackage{enumitem}
\usepackage{cancel}
\usepackage{indentfirst}
\usepackage{graphicx}
\usepackage{tikz-cd}

\newtheorem{theorem}{Theorem}[section] 
\newtheorem{lemma}[theorem]{Lemma}
\newtheorem{proposition}[theorem]{Proposition}
\newtheorem{definition}[theorem]{Definition}

\newtheorem{rem}[theorem]{Remark}

\theoremstyle{definition}

\DeclareMathOperator{\Tr}{Tr}

\setlist[itemize]{label=\textbullet}
\usepackage{float} 
\usepackage[utf8]{inputenc} 
\usepackage{hyperref}
\usepackage{amsmath, amssymb, amsthm} 
\usepackage{eurosym}

\theoremstyle{definition}

\usepackage{wasysym}
\hypersetup{pdfborder={0 0 0}} 
\usepackage{url} 
\usepackage{titlesec}
\titleformat{\paragraph}
{\normalfont\normalsize\bfseries}{\theparagraph}{1em}{}
\titlespacing*{\paragraph}
{0pt}{3.25ex plus 1ex minus .2ex}{1.5ex plus .2ex}

\titleformat{\subparagraph}
{\normalfont\normalsize\bfseries}{\theparagraph}{1em}{}
\titlespacing*{\paragraph}
{0pt}{3.25ex plus 1ex minus .2ex}{1.5ex plus .2ex}

\usepackage{mathrsfs}
\usepackage{leftidx}
\usepackage{braket}

\newcommand{\di}{\mathrm{d}}

\newcommand{\WF}{\mathrm{WF}}

\usepackage{csquotes}
\usepackage[
backend=bibtex,
style=numeric, 
sorting = nyt,
isbn = false,
url = false,
eprint = false,
doi = false,
maxbibnames = 4,
giveninits = true,
date=year,
]{biblatex}
\renewbibmacro{in:}{}
\AtEveryBibitem{%
  \clearfield{number}%
  \clearfield{issue}%
  \clearfield{month}}

\begin{document}
\par 
\bigskip 
\LARGE 
\noindent 
\textbf{Equilibrium states for non relativistic Bose gases with condensation} 
\bigskip \bigskip
\par 
\rm 
\normalsize 
 
\large
\noindent 
{\bf Stefano Galanda$^{1,2,a}$}, {\bf Nicola Pinamonti$^{1,2,b}$}
\\
\par
\small

\noindent$^1$ Dipartimento di Matematica, 
Universit\`a di Genova - Via Dodecaneso, 35, I-16146 Genova, Italy. \smallskip

\noindent$^2$ Istituto Nazionale di Fisica Nucleare - Sezione di Genova, Via Dodecaneso, 33 I-16146 Genova, Italy. \smallskip


\smallskip

\noindent E-mail: 
$^a$stefano.galanda@edu.unige.it, 
$^b$nicola.pinamonti@unige.it
\\

\normalsize
${}$ \\ \\
 {\bf Abstract} \ \ 
In this paper we present the construction of the equilibrium states at finite temperature in the presence of a condensation phase for a Gas of non relativistic Bose particles on an infinite space interacting through a localised two body interaction. 
We use methods of quantum field theory in the algebraic formulation to obtain this result and in order to prove convergence of the partition function and of the generating function of the correlation functions, we introduce an auxiliary stochastic gaussian field which mediates the interaction of the Bose particles  (Hubbard-Stratonovich transformation). The construction of the equilibrium state and of the partition function in the presence of the condensate treating the auxiliary stochastic field as external potential can be achieved using and adapting ideas and methods of Araki.
Explicit formulas for the relative entropy of the equilibrium state with the external potential with respect to the equilibrium state of the free theory are obtained adapting known Feynman-Kac formulas for the propagators of the theory.
If the two-body interaction is sufficiently weak, the proof of the convergence of the partition function after evaluation of the external stochastic field on a suitable Guassian state can be given utilizing the properties of the relative entropy mentioned above.
Limits where the localisation of the two-body interaction is removed are eventually discussed in combination of the limits of vanishing temperature and or in the weakly interacting regime.

\color{black}

\bigskip
${}$

\tableofcontents


\section{Introduction}
\subsection{Overview}
Bose-Einstein condensation (BEC) is a phase transition that occurs at sufficiently low temperatures, where a macroscopic number of bosonic particles accumulate in the lowest energy state of the system \cite{Bose, Einstein}. This phenomenon, in the thermodynamic limit, is characterized by the spontaneous breaking of the global $U(1)$ symmetry, indicated by a non-vanishing expectation value of the bosonic field operator (one-point function) \cite{Bog47, HugenholtzPines, GoldstoneSalamWeinberg, Strocchi, BFP21}. This mechanism underlies key phenomena, such as superfluidity, and provides a fundamental setting for the analysis of emergent behavior in many-body quantum systems.

BEC was first observed experimentally in dilute alkali gases in 1995, confirming long-standing theoretical predictions \cite{And95, Bra95, Dav95}. Following its experimental realization, theoretical research on BEC intensified significantly \cite{PS16}, with notable contributions in the mathematical physics literature aimed at rigorously understanding its underlying mechanisms for the case of particles which interact through a two body force \cite{LiebSeiringerYngvason05, Schlein}. In particular, significant progress has been made in the rigorous analysis of interacting quantum many-body systems at zero temperature, where the excitation spectrum of the Hamiltonian on the ground state \cite{Seiringer, DerNap, BoccatoBrenneckeCenatiempoSchlein, BoccatoBrenneckeCenatiempoSchlein1, BoccatoBrenneckeCenatiempoSchlein2, BOSS, HainzlSchleinTriay}, and the emergence and dynamics of condensation \cite{BreSch, LS02, ErScYa, ErScYa1} have been studied with increasing precision.

The corresponding situation at thermal equilibrium, at positive temperature for weakly interacting particles, has been less thoroughly investigated, though it has attracted growing interest in recent years, particularly with the aim of understanding the behavior of the bosonic gas near the critical temperature. 

Thermal equilibrium in quantum statistical mechanics is characterized by Gibbs states, e.g. grand canonical density matrices specified by the inverse temperature $\beta$ and the chemical potential $\mu$. In particular, to describe the effective behavior of the bosonic gas slightly above the critical temperature, recent works have investigated the emergence of classical Gibbs measures for nonlinear field theories as high-temperature or mean-field limits of grand canonical quantum many-body bosonic systems. These results provide a rigorous derivation of such measures, often associated with nonlinear Schrödinger-type equations, as effective descriptions for bosonic gases with two-body non-local interactions \cite{FroehlichKnowlesSchleinSohinger2, FKSS3, LPR05}. 
In particular, the analysis of Gibbs measure for field theories as limits of certain gran canonical states of the many body systems in the regime of high temperature has been studied in 
\cite{FroehlichKnowlesSchleinSohinger2}, 
where Borel summability of the resulting series expansion has been proven.
The finite temperature case for the two dimensional theory has been discussed in 
\cite{FroehlichKnowlesSchleinSohinger}.
The rigorous derivation of the Gibbs measure as limits of gran canonical states for many body theory, in the case of vanishing background, has been recently presented in \cite{LPR05}.
In all these works, the presence of a finely tuned external trapping potential is assumed together with the absence of a condensate phase in the underlying many-body system. A first step in the direction of including an existing background condensate phase was done in \cite{DNN}. There, trace-norm approximations of the grand canonical Gibbs state are constructed.\\\\

While Gibbs states accurately describe thermal equilibrium in finite or confined systems, their thermodynamic limit typically exhibits divergences. One of the achievements and advantages of the algebraic approach, see \cite{KMS}, was to overcome this problem by defining equilibrium states through the KMS condition, which generalises Gibbs states and provides the correct description of equilibrium in generic infinite-volume settings. The algebraic approach is essential as equilibrium states at the same temperature correspond to inequivalent representations of the same theory; in particular, the existence of multiple, mutually inequivalent KMS states at a fixed temperature is the situation occurring in presence of phase transitions. The definition of perturbed equilibrium systems in statistical mechanics was given in \cite{Araki, BratteliRobinson, OhyaPetz, DJP}, and based on these foundations, extended to quantum field theories became within the algebraic framework \cite{HaagLQP, HaagKastler}. This was achieved through the formalism of perturbative algebraic quantum field theory (pAQFT) \cite{HollandsWald2001, HW02, BDF09, BrunettiFredenhagen00, Fredenhagen2015, KasiaBook, Dutsch:2019aa}, where interacting equilibrium states were constructed order by order in formal perturbation theory \cite{FredenhagenLindnerKMS_2014, DragoHackPinamonti, Galanda}. More recently, techniques from pAQFT have been successfully combined with methods from constructive quantum field theory to construct interacting equilibrium states beyond formal power series, establishing the existence of the theory in \cite{BPR}. Furthermore, in different setups, construction of the theory was obtained for other different models in \cite{BFMThirring, GlimmJaffe, GawedzkiKupiainen, FeldmanMassiveGrossNeveu}.

\subsection{Set up of the problem}
In this paper we consider a system of non relativistic quantum particles satisfying the Bosonic statistics and propagating on $\mathbb{R}^3$. These particles are weakly interacting with a two body force described by a suitable potential, and the model incorporates the presence of a background Bose-Einstein condensate.
We approach this problem using methods of quantum field theory to describe the gas of quantum particles, constructing equilibrium states, following the aforementioned literature, for the interacting system at positive temperature both in the presence and absence of condensation.
The quantum field theory we are considering occupies the whole space, without confinements, and the self-interaction is localised to a compact region.
At least when the interaction among different particles takes place on compact space regions conditions for the existence of the KMS (Gibbs) correlation functions are obtained. Finally, limits where the interaction occupies the whole space are considered. Therefore, using methods of algebraic quantum field theory, in respect to known literature we analyze the construction of the equilibrium states with or without the background taking the large space limits in another way.



The set up is the following. Smeared fields of the theory are denoted by $\Phi(f)$, where $f$ is a generic element of the one particle Hilbert space $L^{2}(\Sigma)$ with $\Sigma=\mathbb{R}^3$. Then, the associated algebra of observables to this field, its adjoint $\Phi^*(\overline{f})$ and its Wick square $|\Phi|^2(h)$, $h\in C^\infty_0(\Sigma;\mathbb{R})$ can be given at fixed time and it is denoted by $\mathcal{F}$. 

The model we will be focusing on is specified by the interacting time evolution that, in a representation of $\mathcal{F}$, is generated by the Hamiltonian 
\[
H=H_0+V
\]
where $H_0$ is the Hamiltonian of the free theory $e^{\mathrm{i}t H_0 }\Phi(f) e^{-\mathrm{i}t H_0} = \Phi(e^{\mathrm{i} t K})$ with associated operator on the one particle Hilbert space
\[
K=-\frac{\Delta}{2m}-\mu
\]
for $m > 0$ the mass, $\mu$ the chemical potential and $\Delta$ the Laplace operator on $\Sigma$. $H_0$ cannot be given within $\mathcal{F}$ (due to infrared divergences), nevertheless, considering the associated time evolution $\tau_t$ ($*$-automorphism on $\mathcal{F}$), equilibrium states at finite inverse temperature $\beta$ are characterized by the Kubo Martin Schwinger (KMS) condition \cite{KMS} and denoted as $\omega^\beta$. 

Instead, the interaction Hamiltonian can be given within $\mathcal{F}$ as
\[
V= \int g |\Phi|^2 v* |\Phi|^2 
\]
here $v \in \mathcal{C}^{\infty}_0(\Sigma, \mathbb{R})$ positive and symmetric is the potential of the two-body force and $g$ is a smooth compactly supported function on $\Sigma=\mathbb{R}^{3}$ (which constrains the support of the interaction to a compact space region). 

Similarly to $H_0$, also $H$ cannot be constructed within the algebra of observables $\mathcal{F}$. In spite of this fact, the cocycle $U(t)$ which intertwines the free and interacting time evolution can be given within $\mathcal{F}$ as the formal time ordered exponential of $V$, which is formally $e^{\mathrm{i}tH}e^{-\mathrm{i}tH_0}$. More explicitly, denoting the interacting time evolution with $\tau^V_t$, we have:
\[
\tau^V_t(A) = U(t)\tau_t(A) U(t)^*.
\]
Even if at this stage the convergence of the power series defining $U(t)$ is not under control, 
we may use ideas similar to those given in \cite{Araki, FredenhagenLindnerKMS_2014} to construct the equilibrium state relative to the interacting evolution. By proceeding in this way, the relative partition function of the interacting equilibrium is obtained by means of an analytic continuation of $\omega^\beta(U(t))$ as $Z= \omega^{\beta}(U(\mathrm{i}\beta))$ (formally $\Tr(e^{-\beta H})/\Tr (e^{-\beta H_0})$). This is given again as a series in $\mathcal{F}$ in powers of the interaction Hamiltonian $V$ and the analytic continuation can be taken at all orders in perturbation theory thanks to the KMS property of $\omega^\beta$. The convergence of this power series needs and is carefully discussed in this paper. 

In analyzing this problem, it becomes immediately clear that a new technique based on constructive methods must be developed. Indeed, upon initial inspection, the $N$-th order term in the power series defining the relative partition function is equivalent to a sum over all possible graphs with $N$ vertices. However, in this sum, the corresponding number of diagrams that needs to be estimated grows, at order $N$, as $(2N)!$ (see \cite{Cvitanovic}, or Appendix \ref{ap:counting}). 
Hence, even though the power series defining $\omega^\beta(U(\mathrm{i}\beta))$ has a prefactor $1/N!$, we cannot achieve absolute convergence in this way. 

\subsection{Strategy adopted}
To prove that the relative partition function $Z=\omega^{\beta}(U(\mathrm{i}\beta))$ and the corresponding correlation functions can be given in terms of a convergent series we proceed as follows.
We introduce an auxiliary stochastic Gaussian field $A$, and we decompose the quartic non-local interaction Hamiltonian $V$ as a product of two potentials  quadratic in the field $\Phi$ and linear in $A$
\[
Q_A = \int \di^3 x \, g |\Phi|^2 A.
\]
In the case of a non vanishing background the field decomposes as $\Phi=\phi_0+\Psi$, with $\phi_0$ the classical part describing the condensate and $\Psi$ quantum fluctuations on it. Then, the at most quadratic potential considered is
\[
Q_A = \int \di^3 x \, g (|\Psi|^2 +\phi_0(\Psi+\Psi^*))A.
\]
Furthermore, the covariance of the Gaussian field is chosen to match the non-local interaction potential. 
This trick, referred to as Hubbard–Stratonovich (HS) transformation, is well known in the literature, see e.g. \cite{Rivasseau, RivasseauBook}, and it has already been used in a similar context in \cite{FroehlichKnowlesSchleinSohinger2}. Here we want to push its use further ahead, to analyze the construction of KMS (Gibbs) states for the case of non vanishing background, when a condensate is present.

The first step consists in treating $A$ as external potential. This can be done directly, because the corresponding theory is linear in the quantum field (quadratic Hamiltonian), leading to a partial re-summation of the original series. 
Then, in order to get the correct expression for the relative partition function $Z=\omega^\beta(U(\mathrm{i}\beta))$ (up to an inessential constant), we need to evaluate the external field in a Gaussian state with covariance $v$. We actually prove in Proposition \ref{prop:auxiliary-field} and in Proposition \ref{prop:auxiliary-partition-function}, a priori just as formal power series, that
\[
\omega^{\beta }(U(\mathrm{i}\beta)) = \mathbb{E}(\omega^\beta (U_A(\mathrm{i}\beta))).
\]
Here $Z_A=\omega^\beta (U_A(\mathrm{i}\beta))$ is the relative partition function for the quadratic perturbation $Q_A$ and $\mathbb{E}$ denotes to the evaluation of $A$ in a Gaussian state with suitable covariance.

At this stage, the relative partition function of the auxiliary theory $U_A(\mathrm{i}\beta)$ still depends in a complicated way from $A$. Hence, the naive evaluation on the Gaussian state of the power series corresponding to $\omega^\beta(U_A(i\beta))$ still leads to a number of contribution that is not yet summable. In order to overcome this last problem we decompose $U_A(\mathrm{i}\beta)$ and the relevant object in the theory in sum of powers of exponential of $A$. 
In this way, the combinatoric arising when the evaluation in the auxiliary Gaussian state is realized, is simpler and directly under control. 
In this process it plays an important role the observation that $\log(\omega^\beta(U_A(\mathrm{i}\beta)))$ is related to the relative entropy of the equilibrium state $\omega^{\beta A}$ of the time evolution perturbed with $Q_A$ relative to the free equilibrium state $\omega^{\beta}$ (see e.g. \cite{OhyaPetz, ArakiEntropy}). Therefore, as part of this analysis, an explicit expression for the relative entropy is obtained similar to the results in \cite{BruFrePin25}.

\subsection{Main result}

The considered system of interacting particles is described by the set of physical parameters $\phi_0, \beta, v, g$ where $\phi_0$ is the background part of the field $\Phi$ corresponding to the condensate. $\beta$ is the inverse temperature of the considered thermal equilibrium state, $v$ is the potential of the repulsive two-body force among particles and $g$ a cutoff function which localises the interaction on a compact region. The main result of this paper is Theorem \ref{th:convergence}, where we prove that, for the above set of physical parameters satisfying certain bounds, the relative partition function and the correlation functions converge. More precisely, convergence is implied by the following (non optimal) constraint 
    \[
    1+\exp\left(- \frac{\phi_0^4}{2} \beta v(0) \|g\|_1^2\right) > \exp\left(\frac{ \sqrt{\beta v(0) \|g\|_1^2}}{{\beta^{\frac{3}{2}}}}  \tilde{{C}}\right)
    \]
where $\tilde{C}$ is a given numerical constant. We see in particular that when $\phi_0, v, g$ are kept fixed there is a range of inverse temperature $\beta \in (\beta_0,\beta_1)$ where existence of the partition function of the system is implied by Theorem \ref{th:convergence}. We furthermore notice that when $\phi_0=0$, namely when no background is considered, this range increases.

It is furthermore possible to observe that limits where some of this parameter diverge keeping convergence of the partition function can be taken. 
In particular, denoting $\beta v(0) \|g\|_1^2$ by $\gamma$, in the limit where $\|g\|_1\to \infty $ and $ \beta v(0) \to 0$ and
keeping $\gamma$ constant we have that we get convergence of the partition function and of the correlation function if 
$\beta \in (\beta_0,+\infty)$ where here 
\[
\beta_0= \left(\frac{\gamma}{\tilde{C}} \log(1+e^{-\gamma \phi_0^4/2}) \right)^{\frac{3}{2}}
\]
Although these bounds are not optimal, this behavior of $\beta_0$ as a function of $\phi_0^2$, the condensate density, reminds  the form of the critical temperature as a function of the condensate density.


\subsection{Structure of the paper}

In the next Section we present the algebra of the field theory we are considering.
We remind how to enlarge the observable algebra to include normal ordered Wick squares, we recall the form of the free equilibrium state and how to obtain expressions for the interacting equilibrium state and for the corresponding relative partition function.

In Section \ref{se:3} we present the auxiliary theory with the external gaussian field $A$ and how to connect it to the problem of the construction of equilibrium state of the interacting theory. We furthermore recognise that the logarithm of the partition function is controlled by the expectation value of the interaction Hamiltonian and the relative entropy of the equilibrium state of the auxiliary theory relative to the free equilibrium state.
In section \ref{se:4} we construct the propagators of the auxiliary theory and we give expressions for the relative partition function. 
Section \ref{se:5} contains the discussion of the path integral representation of the propagators of the auxiliary theory.
In Section \ref{se:6} we show how to evaluate the external gaussian field $A$ to recover the relative partition function and the correlation function of the theory, we furthermore present the theorems ensuring convergence of this procedure.  

\section{Algebraic relations}

Following \cite{ArakiWoods}, we start considering the CCR algebra of free non-relativistic quantum fields.
The fields we shall consider are smeared with test function defined on the three dimensional space which we denote by $\Sigma = \mathbb{R}^{3}$. The definition is the following.

\begin{definition}\label{def:free-abstract-algebra} The {\bf algebra of free non-relativistic fields} $\mathcal{A}$ is the $*$-algebra generated by the identity $\mathbb{1}$ and the linear complex bosonic scalar fields $\Phi(h)$, $\Phi^*(f)$ smeared with 
Schwartz functions $f, h\in \mathcal{S}(\Sigma;\mathbb{C})=\mathcal{S}$,
%
that satisfy:

\begin{itemize}
	\item {\bf Hermicity:} The $*$-involution acts for any 
$f\in \mathcal{S}$ as: 
\begin{equation}\label{eq:operation-*}
(\Phi(f))^{*} =
\Phi^{*}(\overline{f}), \qquad 
(\Phi^{*}(f))^{*} =
\Phi(\overline{f}).
\end{equation}
	\item {\bf Canonical Commutation Relations:} For any 
$f, h\in \mathcal{S}$:
\begin{equation}\label{eq:commutation-relations}
    [ \Phi(h),\Phi^*(f)] = \langle h,f \rangle = \int_{\mathbb{R}^3} \overline{h}f\;  \di^3 x,
    \qquad
    [\Phi(f), \Phi(h)] = 0
    \qquad
    [\Phi^*(f), \Phi^*(h)] = 0.
\end{equation}
\end{itemize}

\end{definition}
Notice that since Schwartz function $\mathcal{S}(\mathbb{R}^3;\mathbb{C})$ are dense in the {\bf one particle Hilbert space} $\mathcal{H}=L^2(\mathbb{R}^3,\mathbb{C})$, the above definition of the free algebra can be completed to include generators smeared with elements of $\mathcal{H}$. Furthermore, in the following, we shall often use the following formal notation
\begin{equation}\label{eq:linear-fields-integrals}
\Phi(h) = \int \overline{h}(x) \Phi(x) \di^3 x,
\qquad 
\Phi^*(f)  
=
\int \Phi^*(x)f(x) \di^3 x,
\end{equation}
to indicate the generators, which highlights that the labels of the abstract fields are elements of the one-particle Hilbert space $\mathcal{H}$.
 
%
%
%

%

The first relevant example of state on $\mathcal{A}$, following again \cite{ArakiWoods}, is the {\bf vacuum state}, denoted by $\omega^\infty$. This is a quasi-free state on $\mathcal{A}$, which is constructed out of the following two-point functions
\[
\omega^\infty(\Phi(h)
\Phi^{*}(f)) = \langle h,f \rangle
,
\qquad
\omega^\infty(
\Phi^{*}(f)\Phi(h)) = 0
,
\qquad
\omega^\infty(\Phi(f)
\Phi(h)) =  0, 
\qquad
\omega^\infty(\Phi^{*}(f)
\Phi^{*}(h)) =  0
\]
from which one can see that the associated GNS representation is on a Fock space constructed over the one-particle Hilbert space $\mathcal{H}$ of before. Moreover, on this Fock space, the field is represented just by an annihilation operator while its adjoint by a creation operator.

\subsection{Time evolution and equilibrium state}

We specify the time evolution on $\mathcal{A}$, necessary to determine the correlation functions of the free equilibrium states at positive temperature. 
To this end, consider a self-adjoint operator ${K}$ on the one particle Hilbert space $\mathcal{H}={L^2}(\Sigma;\mathbb{C})$, with $\Sigma = \mathbb{R}^3$, that generates the free time evolution. More precisely, the operator we consider has the following form  
\begin{equation}\label{eq:def-K}
K \coloneqq \mathcal{K}_{\mu} = \mathcal{K}-\mu = -\frac{\Delta}{2m} - \mu
\end{equation}
where $m>0$ is the mass of the particle, $\mu$ the 
chemical potential and $\Delta$ the Laplace operator on $\Sigma$ with the convention that $-\Delta$ is positive.
We consider the case where the spectrum of ${K}$ is contained in $[\epsilon,+\infty)$ for $\epsilon >0$. This means that $\mu<0$ (eventually, we shall take the limit where $\epsilon \to 0$ or similar limits). 
By Stone theorem, $K$ generates a strongly continuous one parameter group of unitary transformations $e^{i t {K}}$, in the parameter $t \in \mathbb{R}$, which is used to describe the free time evolution
\[
f(t,x) \coloneqq f_t(x) \coloneqq e^{i t {K}} f(x) \qquad \forall f \in \mathcal{H}. 
\]
Raising this definition at the algebraic level gives a corresponding one parameter group of automorphisms $\tau_t$ of $\mathcal{A}$, whose action on the generators is
\begin{equation}\label{eq:time-evolution}
\tau_t(\Phi^*(f)) = \Phi^*(e^{\mathrm{i}tK}f),
\qquad
\tau_t(\Phi(f)) = \Phi(e^{\mathrm{i}tK}f).
\end{equation}
The operator $K$ and the corresponding automorphisms $\tau_t$ play the role of the free Hamiltonian of the system and that of the automorphism implementing the free time evolution.
Both the commutation relations and the expectation values in the vacuum state are left invariant under the action of $\tau_t$. 
Given $\tau_t$, the corresponding {\bf equilibrium states} $\omega^\beta$ at inverse temperature $\beta$ are the quasi-free states characterized by the following two-point functions
\begin{equation}\label{eq:omegabeta}
\begin{aligned}    \omega^\beta(\Phi(f)
\Phi^{*}(h)) &= \langle f,B_- h \rangle ,
\\
\omega^\beta(
\Phi^{*}(h)\Phi(f)) &= \langle f,B_+h \rangle
\\
\omega^\beta(\Phi(f)
\Phi(h)) &=  0, 
\\
\omega^\beta(\Phi^{*}(f)
\Phi^{*}(h)) &=  0
\end{aligned}    
\end{equation}
where the Bose factors, given in terms of $K$, were denoted by
\[
B_{\pm} \coloneqq \mp \frac{1}{1-e^{\pm\beta K}}.
\]
Notice that, for $\mu<0$ the spectrum of $K$ is positive and supported away from $0$. Hence, $B_\pm$ are bounded operators on $\mathcal{H}=L^2(\mathbb{R}^3;\mathbb{C})$.  

%
%

\subsection{Normal ordering, Wick squares and enlarged Algebra}
In order to describe the interaction Lagrangian of the two-body force among the bosonic fields at an algebraic level we need to enlarge the algebra we are working with. We shall realise this by adding to the set of generators the square, normal ordered, fields.
We recall that the normal ordering of quadratic fields is obtained subtracting the expectation value on the vacuum. In particular, the following formal formula is employed
\begin{equation}\label{eq:formal-product}
\begin{aligned}
:\!\Phi^*\Phi\!:\!(f) 
&\coloneqq
\int 
\left(
\Phi^*(x)\Phi(y) - \omega^\infty(\Phi^*(x)\Phi(y)) \right)f(x,y) \di^3 x \di^3 y
\\
:\!\Phi\Phi^*\!:\!(f)
&\coloneqq
\int 
\left(
\Phi(y)\Phi^*(x) - \omega^\infty(\Phi(y)\Phi^*(x)) \right)f(y,x) \di^3 x \di^3 y
\\
&\coloneqq
\int 
\left(
\Phi^*(x)\Phi(y) - \omega^\infty(\Phi^*(x)\Phi(y)) \right)f(y,x) \di^3 x \di^3 y\; 
\end{aligned}
\end{equation}
where the object used to smear these elements is $f\in \mathcal{S}'(\mathbb{R}^3\times\mathbb{R}^3)$ with a constraint on its singularity. Namely, we demand that $f=f_s \circ e^{\mathrm{i}t_1K}\otimes e^{\mathrm{i}t_2 K}$ for some $t_1,t_2 \in \mathbb{R}$ with $f_s\in \mathcal{E}'(\mathbb{R}^3\times\mathbb{R}^3)$ and the {\bf wavefront set} \cite{HormanderI} of $f$ is such that $\WF(f)\subset  Y := \{(x,y;k_x,k_y)\subset T^*\Sigma^2 \, |\, x=y, k_x=-k_y\neq 0\} =\WF(\delta)$ where $\delta$ is the Dirac delta distribution supported on the diagonal of $\mathbb{R}^3\times \mathbb{R}^3$. We denote this {\bf set of labels} by $E$.

We observe that $:\!\Phi^*\Phi\!:\!(f)=:\!\Phi\Phi^*\!:\!(sf)$, with $sf(x,y) = f(y,x)$. Furthermore, if $\WF(f)=\emptyset$, namely if
$f$ is a smooth function,  $:\!\Phi^*\Phi\!:\!(f)$ can be obtained as sum of products of $\Phi(f_1)$, $\Phi(f_2)$ and the identity, up to a completion which is the same used to obtain $C^{\infty}_0(\mathbb{R}^3\times\mathbb{R}^3)$ as sum of elements of  $C^{\infty}_0(\mathbb{R}^3)\times C^{\infty}_0(\mathbb{R}^3)$. Finally, notice that the Wick square is contained within this class of objects choosing 
\[
:\!\!|\Phi|\!\!:^2(h) = :\!\Phi\Phi^*\!:\!(h\delta) = :\!\Phi^*\Phi\!:\!(h\delta) 
\]
where $h\in{C}_0^\infty(\mathbb{R}^3)$.
%
%
%
%
%

The observables obtained in this way are only formal. In particular, the coinciding point limits acquire direct  meaning only if they are taken after evaluation on suitable states. 
In spite of this problem, the commutation relations of normal ordered local fields with other elements of the algebra can be analyzed. See \cite{HollandsWald2001,HW02} for a discussion in the case of relativistic fields. 
The obtained local fields can thus be used together 
with the linear fields as generators of an extended algebra of fields, presented in the following definition.  
\begin{definition} \label{def:extended-algebra}
The {\bf abstract extended algebra} $\mathcal{F}$ is the $*$-algebra generated by the identity $\mathbb{1}$, the linear complex bosonic scalar fields $\Phi(h)$, $\Phi^*(g)$ smeared with elements of the one-particle Hilbert space $g, h\in L^2(\mathbb{R}^3;\mathbb{C})$
and by the normal ordered products $:\!\Phi^*\Phi\!:^2(f)$,  $:\!\Phi\Phi^*\!:^2(f)$ which are smeared with $f\in E:= \{f=f_s \circ e^{\mathrm{i}t_1K}\otimes  e^{\mathrm{i}t_2K} | f_s \in \mathcal{E}'(\mathbb{R}^3\times \mathbb{R}^3), t_1,t_2 \in\mathbb{R},  \WF(f)\subset \WF(f_s)\subset Y \}$, with $Y=\WF(\delta)$.
The product among these objects and the $*$-operation is obtained from the formal definition given in equation \eqref{eq:formal-product} and the corresponding product in $\mathcal{A}$ as given in Definition \ref{def:free-abstract-algebra}.

%
%
%
\end{definition}
We observe that $\mathcal{A}\subset \mathcal{F}$. 
Furthermore, the Wick square $:\!\!|\Phi|\!\!:^2(f)$ with $f\in C^\infty_0(\mathbb{R}^3)$ is contained in $\mathcal{F}$. 
The product in $\mathcal{F}$ is non commutative and the non commutative part of the product is read from the commutation relations enumerated in Definition \ref{def:free-abstract-algebra} together with the formal definition of normal ordered products given in equation \eqref{eq:formal-product}. In particular, we have:
\begin{equation} \label{eq:extended-commutation}
    \begin{aligned}
    [:\!\!|\Phi|\!\!:^2(f),:\!\!|\Phi|\!\!:^2(h)] &= 0, \qquad \qquad \, f,h \in C^\infty_0(\mathbb{R}^3)\\
    [:\!\!|\Phi|\!\!:^2(f),\Phi(h)] &= \Phi(\overline{f}h) \qquad \, \, f \in C^\infty_0(\mathbb{R}^3), h\in \mathcal{S}
    \\
        [:\!\!|\Phi|\!\!:^2(f),\Phi^*(h)] &= \Phi^*({f}h) \qquad f \in C^\infty_0(\mathbb{R}^3), h\in \mathcal{S}
    \end{aligned}
\end{equation}
and the $*$-operation acts on $:\!\!|\Phi|\!\!:^2(f)$ in the following way
\begin{equation}\label{eq:extended-*}
    (:\!\!|\Phi|\!\!:^2(f))^* =
    :\!\!|\Phi|\!\!:^2(\overline{f}).
\end{equation}
Similarly, the $*$-automorphisms describing the time evolution originally given in $\mathcal{A}$ can be extended to $\mathcal{F}$, its action is given only implicitly through the formal definition of the added generators described by the normal ordered fields.

We shall present below a concrete realization of this algebra.

\color{black}

\subsubsection{Extended algebra, normal ordering and the state}

The abstract extended algebra introduced in Definition \ref{def:extended-algebra} is generated by formal objects, like the normal ordered Wick square, which acquires meaning only when the evaluation on a suitable state is considered. With ideas developed in the context of pAQFT \cite{BDF09, Fredenhagen2015, KasiaBook}, we can find a faithful representation of this algebra in terms of concrete and finite objects. We actually represent the extended algebra $\mathcal{F}$
with a product adapted to the state we are considering.
\begin{definition}
The {\bf extended thermal representation} is denoted by $\mathcal{F}_\beta$ and it is the $*$-algebra formed by the smaller set of functionals over off-shell field configurations $\Phi \in C^{\infty}(\mathbb{R}^3)\cap \mathcal{S}'(\mathbb{R}^3)$ and $\Phi^*\in C^{\infty}(\mathbb{R}^3)\cap \mathcal{S}'(\mathbb{R}^3)$, which is generated by the identity and the following 
\[
\Phi(f)= \int   \overline{f} \Phi(x)\di^3 x, \qquad \Phi^*(f)= \int   {f} \Phi^*(x)\di^3 x, \qquad 
\Phi^*\Phi(h) = \int \Phi^*(x) \Phi (y) h(x,y) \di^3 x \di^3 y 
\]
where $f\in \mathcal{S}(\mathbb{R}^3)$, $h\in E$ as given in Definition \ref{def:extended-algebra}. The product is defined as
\begin{equation}\label{eq:star-product}
A\star_\beta B = Me^{D_{12}}
A\otimes B, \qquad A,B\in\mathcal{F}_\beta
\end{equation}
where $M(A\otimes B)= AB$, pointwise product, and
\begin{equation}\label{eq:Dij}
D_{ij} = \langle   \omega^\beta_{2,-},\frac{\delta}{\delta \Phi_i} \otimes \frac{\delta}{\delta \Phi_j^*{}}
\rangle
+
\langle  \omega^\beta_{2,+}, \frac{\delta}{\delta \Phi_i^*} \otimes \frac{\delta}{\delta \Phi_j{}} \rangle
\end{equation}
with functional derivatives $\delta/\delta \Phi_i$ acting on the $i$-th factor of the tensor product.
Furthermore, $\omega_{2,\pm}^{\beta}$ are the distributions (integral kernels) associated to the operators $B_{\pm}$. The $*$-operation in $\mathcal{F}_\beta$ corresponds to the complex conjugation, and it maps $\Phi(f)$ to $
\Phi^*(\overline{f})$. 
\end{definition}
As we shall prove in the next proposition, in order to ensure that $\mathcal{F}_\beta$ is a well defined algebra, it is useful to expand the $\star_\beta$ product among $n$ factors in terms of a sum over simple (unoriented) graphs among $n$ vertices. 
We recall that a graph $G$ is formed by a set of vertices $N(G)$ (points) and by a set of edges $E(G)$ (lines connecting the vertices). Each edge $e=\{p,q\}\subset E(G)$ connects two endpoints $p,q\in V(G)$. 
Furthermore, a graph is simple if it has no loops nor multiple lines joining the same couple of vertices. 
We say that a graph contains a loop if at least one of its edges starts and ends in the same vertex. 
We furthermore denote the general set of these graphs with $\mathcal{G}_n$, notice that it contains also the graph with no edges. 
The subset $\mathcal{G}_n^c \subset \mathcal{G}_n$, is the subset of simple connected graphs. Recalling equations \eqref{eq:star-product} and the definition of $D_{ij}$ given in equation \eqref{eq:Dij}, we have that
\[
A_1\star_\beta \dots \star_\beta A_n
=Me^{\sum_{i<j} D_{ij}}
A_1\otimes \dots \otimes  A_n , \qquad 
A_i\in \mathcal{F}_\beta
\]
where the $i$-th and $j$-th functional derivative present in $D_{ij}$ act respectively on $A_i$ and $A_j$ and $M(A_1\otimes\dots \otimes A_n ) = A_1\cdots A_n$.
We can now obtain an expansion of the previous formula as a sum over graphs.
\begin{equation}\label{eq:graph-decomposition}
\begin{aligned}
A_1\star_\beta \dots \star_\beta A_n
&= 
 M\left(\prod_{i<j}(e^{D_{ij}}-1+1)\right)
A_1\otimes \dots \otimes  A_n
\\
&= 
M \sum_{G\in \mathcal{G}_n}  \left(\prod_{\{i,j\}\in E(G)\atop i<j}(e^{D_{ij}}-1)\right)
A_1\otimes \dots \otimes  A_n , 
\end{aligned}
\end{equation}
where in the second equality we expanded the product of the sums of  $(e^{D_{ij}}-1)$ with $1$ and we have observed that it results in a sum over all possible graphs joining the vertices $\{1,\dots, n\}$. Notice that, evaluating on a state, the vertices of each graph are in correspondence with entries of the $n$-point function (ordered in the same way) and the edges to $e^{D_{ij}}-1$ (ordered in ascending order, namely $i<j$).

The contribution to every $G \in \mathcal{G}_n$, given in the expansion in equation \eqref{eq:graph-decomposition}, can be further decomposed into a sum over the connected components of $G$. Rearranging the sum we thus obtain the following decomposition
\begin{align}\label{eq:decomposition-connected}
A_1 \star_\beta \dots \star_\beta A_n 
&= 
\sum_{\pi\in \mathcal{P}_n}\prod_{I\in \pi}
\left(
 M\sum_{G\in \mathcal{G}^c(I)}  
\left( \prod_{\{i, j\}\in E(G) \atop i<j}(e^{D_{ij}}-1)\right)
 A_{I_1}\otimes \dots \otimes  A_{I_{|I|}} \right), 
\end{align}
where $\mathcal{P}_n$ is the set of partitions of $\{1,\dots, n\}$ and
$\mathcal{G}^c(I)$, with $I\in \pi\in \mathcal{P}_n$, is the set of simple connected graphs among the vertices contained in $I$ and as above the $i$-th and $j$-th functional derivative present in $D_{ij}$ act respectively on $A_i$ and $A_j$. We can now present and prove the following proposition ensuring the well posedness of the construction.

\begin{proposition}
The product in $(\mathcal{F}_\beta,\star_\beta)$ is well defined.   
\end{proposition}
\begin{proof}
To prove that the product in $\mathcal{F}_\beta$  is well defined, we consider a generic element of $\mathcal{F}_\beta$ obtained multiplying various generators of the form
\[
   F=A_1 \star_\beta\dots \star_\beta A_n
\]
with
\[
   A_i\in\{\Phi(f),\Phi^*(\overline{f}),|\Phi|^2(f)| f\in C^{\infty}_0(\Sigma, \mathbb{C})\}.
\]
The extension to the case where $\Phi(h)$ and $\Phi^*(h)$ are smeared with Schwartz functions
and to the case where $\Phi^*\Phi$ is smeared with an element of $E$ ,given in Definition \ref{def:extended-algebra}, can be done by completion arguments.
%
%
%
Recalling equation \eqref{eq:decomposition-connected} we can decompose the $\star_{\beta}$-product present in $F$ as a sum over simple connected graphs whose edges are in correspondence to $e^{D_{ij}}-1$ and whose vertices to $A_i$.
%
%
%
	The factors $A_i$ considered here are at most quadratic in the fields, hence the diagrams in $\mathcal{G}^c(I)$ which have non vanishing contributions in $F$ have at most two edges starting or ending in a single vertex. For this reason, there are at most three types of simple connected diagrams which contributes to $F$:
    The trivial diagrams formed by a single vertex and no edges;
    Those which have $|I|-1$ edges in which there are two vertices reached by a single edge (we call them big lines) and the other $|I|-2$ vertices reached by two edges.
    Finally, the third set of diagrams is formed by those graphs which have $|I|$ edges corresponding to a big loop.
    
    For the same reason, in the expansion of $\prod_{i<j}(e^{D_{ij}}-1)$ in powers of $D_{ij}$, for each diagram with more than two vertices only the first order has non vanishing contribution; while for the diagram with one edge and two vertices we have only two contributions $D_{ij}$ and $D_{ij}^2$. 
   
    Consider now the distributions $\omega_{2,\pm}^\beta$ present in $D_{ij}$. We observe that the integral kernel of $\omega_{2,+}^\beta$ is actually a smooth function while $\omega_{2,-}^\beta$ is characterized by a singularity similar to the one present in the two-point function of the vacuum. Actually, $\omega^\beta_{2,-}-\omega_{2,-}^\infty$ is also a smooth function and $\omega_{2,-}^\infty(x,y) = \delta(x-y)$. Hence, in terms of their wave front set $\WF(\omega^{\beta}_{2,+})=\emptyset$ and $\WF(\omega^{\beta}_{2,-})=Y=\{(x,k_x;y,k_y)\in T^*\Sigma^2| x=y, k_x=-k_y\neq 0\}$.
    
    With this structure, and thanks to the compact support of $A_i$ for all $i$, we observe that the operations in terms of $\omega_{2,\pm}^\beta$ associated to big lines with more than two vertices is actually a composition of distributions with integral kernel $\omega_{2,\pm}^\beta$ which is then applied to the test functions associated to the end points of the big lines. These operations are always well defined thanks to the singular structure discussed above. More precisely, in each single vertex $i$ reached by two edges we are actually considering a composition of distributions of the form  
    \[
    \sum_{\{\pm,\pm\}} \omega^\beta_{2,\pm} \circ f \omega^\beta_{2,\pm},
%
    \]
    for some compactly supported smooth function $f$ which is the smearing function of $|\Phi|^2(f)=A_i$.
    To prove that this is a well defined distribution we use Theorem 8.2.14 in \cite{HormanderI}.
    The very same Theorem implies that $\omega^\beta_{2,+}\circ f\omega^\beta_{2,\pm}$ and $\omega^\beta_{2,\pm}\circ f\omega^\beta_{2,+}$ are distribution with smooth integral kernel and 
    that
    \[\WF(\omega^\beta_{2,-}\circ f \omega^\beta_{2,-}) \subset \WF(\omega^\beta_{2,-}).
    \]
Hence with the same reasoning we have that  multiple compositions are well defined too.

    Similarly, in the expansion in terms of $\omega^\beta_{2,\pm}$, associated to big loops, we observe that at least one edge is associated to $\omega^{\beta}_{2,+}$ which is smooth. Hence the  operation in terms of distributions corresponding to big loops is always well defined.
    It correspond to the distributions discussed for the big lines tested on $\omega^\beta_{2,+}$ after localisation with a smooth function on suitable compact spacetime region.

    The term corresponding to  the diagram with two vertices and one edge which is proportional to  one single $D_{ij}$ can be treated as the big lines while the contribution which is proportional to $D_{ij}^2$ can be treated as a big loop. The factors corresponding to trivial diagrams are not joined with any distributions hence there are no distribution to consider.
    This concludes the proof.
\end{proof}

The map which realises the representation of $\mathcal{F}$ to $\mathcal{F}_\beta$ is the inverse of
\begin{equation}\label{eq:alphabeta}
\alpha_{\beta}^{-1} := e^{-\frac{1}{2}\int \left(\omega^{\beta}_{2,+}(y,x)+\omega^{\beta}_{2,-}(x,y)\right) \di^3 x \di^3 y  \frac{\delta^2}{\delta \varphi(x) \delta{\varphi^*(y)}}}
\end{equation}
where
$\omega^{\beta}_{2,\pm}(x,y)$ are the two-point functions of the free equilibrium state. 

Notice in particular that, while $\alpha_\beta(\Phi) = \Phi$ and $\alpha_\beta(\Phi^*) = \Phi^*$, the normal ordered product of fields satisfies an affine transformation law under the application of $\alpha_\beta$. Actually
\begin{equation}\label{eq:ren-finite-Wick-square}
    \alpha_\beta :\!\!|\Phi|\!\!:^2 = |\Phi|^2 +  \int \di^3p \frac{e^{-\frac{\beta p^2}{2m}+\beta \mu} }{1-e^{-\frac{\beta p^2}{2m}+\beta \mu}} = |\Phi|^{2}+c_\beta
\end{equation}
where, for $\mu\leq 0$, $c_\beta$ is a suitable finite number. We observe that $\alpha_\beta:\mathcal{F}\to(\mathcal{F}_\beta,\star_\beta)$ is a $*$-homomorphisms of algebras that, for every $A, B \in \mathcal{F}$, is also an algebra $*$-isomorphism
\[
\alpha_\beta (AB) = \alpha_\beta A \star_\beta \alpha_\beta B, \qquad \alpha_\beta (A^*) = (\alpha_\beta A)^*.
\]
Furthermore, the free equilibrium state $\omega^\beta$ extended on $\mathcal{F}_\beta$ takes a particularly simple form.
Namely for every $F\in \mathcal{F}$ it holds that 
\begin{equation}\label{eq:eval-state-beta}
\omega^\beta(F) = \left.\alpha_\beta F \right|_{\Phi=0,\Phi^*=0}.
\end{equation}
Namely, it is 
the push forward of $\omega^\beta$ to $(\mathcal{F}_\beta,\star_\beta)$ under $\alpha_\beta$, consisting to the evaluation on the vanishing field configurations.

In conclusion, when working with $\mathcal{F}_\beta$ we are dealing with finite, concrete objects and the free equilibrium state takes a particularly simple representation. The drawback is that the product we are considering depends on the state and the normal ordered fields acquire an extra finite renormalization depending on the temperature and on the chemical potential, see $c_\beta$ in \eqref{eq:ren-finite-Wick-square}.

\color{black}


\subsection{Interacting time evolution}\label{se:interacting-time-evolution}
Consider the equilibrium state $\omega^\beta$, extended on $\mathcal{F}$, introduced above in \eqref{eq:omegabeta}.
The interacting time evolution is obtained perturbing the free time evolution with the interaction  Hamiltonian which is described by the following formally self adjoint element of $\mathcal{F}$
\begin{equation}\label{eq:interaction-hamiltonian}
\begin{aligned}
V &:= \frac{1}{2}\int :\!\Phi(x)\Phi^*(x) \Phi(y)\Phi^*(y)\!: v(x-y) g(x) g(y) \di^3 x\di^3 y
\\
 &= \frac{1}{2}\int :\!\!|\Phi|\!\!:^2(x) :\!\!|\Phi|\!\!:^2(y) v(x-y) g(x) g(y) \di^3 x\di^3 y
 -\int :\!\!|\Phi|\!\!:^2(x)  v(0) g(x) g(x) \di^3 x + c
\end{aligned}
\end{equation}
where $g\in C^\infty_0(\mathbb{R}^3;\mathbb{R}^+)$ is an infrared cutoff function, positive and equal to $1$ on large regions of space and $v \in C^{\infty}_0(\mathbb{R}^3;\mathbb{R}^+)$ is the symmetric potential of the two-body force among the particles forming the field theory we are considering. In this work, we shall also assume that the two-point function constructed with $v$
\[
{v}_2(f,g)  := \int \di^3 x \di^3 y f(x)v(x-y)g(y)
\]
defines a positive product on $\mathcal{D}(\Sigma)={C}^\infty_0(\Sigma, \mathbb{C})$.
Furthermore, in $V$, $c$ is an inessential constant. In the limit where $g\to 1$, the second contribution at the right hand side of the second line of equation \eqref{eq:interaction-hamiltonian}, which appears because of normal ordering, can be reabsorbed in a redefinition of the chemical potential $\mu$. 
The properties required for the two-body force potential $v$, make $V$ formally self adjoint and are going to be crucial to get convergence of the perturbative construction of the equilibrium state. As a consequence, in the GNS representation of $\omega^\beta$, $V$ is described by a self-adjoint operator.\\\\

Therefore, the interaction Hamiltonian of the system we are considering is described by $V$ and
%
we use it to perturb the free dynamics described by the automorphisms $\tau_t$ generated by $K=\mathcal{K}-\mu$ and whose action on linear fields is given in equation \eqref{eq:time-evolution}. Indeed, the dynamics perturbed by the interaction Hamiltonian $V$, is described by the algebraic automorphisms $\tau^V$ and it is such that for any $t  \in \mathbb{R}$ 
\[
\tau^V_t (A) = U(t)\tau_t (A) U(t)^*, \qquad A \in \mathcal{F}
\]
where $U(t)$ is a cocycle which satisfies
\[
U(t+s)= U(t)\tau_t(U(s)).
\]
Its infinitesimal generator is $V$, namely
\[
\frac{d}{\di t}U(t) = \mathrm{i} U(t)\tau_t(V).
\]
$U(t)$ given above is only a formal object which intertwines the automorphisms $\tau$ and $\tau^V$. Nevertheless, we observe that these automorphisms can be implemented in the GNS representation of various state. In particular, if we consider $(\mathcal{H},\pi,\psi)$, the GNS representation of $\omega^\beta$ and if we denote by $H_0$ the self adjoint operator which implements $K$ in $(\mathcal{H},\pi,\psi)$,
we have that 
\[
\pi (U(t)) = e^{\mathrm{i} t (H_0+\pi (V))}e^{-\mathrm{i} t H_0}.
\]
Where we denoted by $H_0+\pi (V)$ the operator which implements the interacting time evolution. 
Furthermore, the action of the time evolution gets represented as
\[
\pi (\tau_t (A)) = e^{\mathrm{i}t H_0} \pi (A) e^{-\mathrm{i}t H_0}, \qquad
\pi (\tau^V_t (A)) = e^{\mathrm{i}t (H_0+\pi (V))} \pi (A) e^{-\mathrm{i}t (H_0+\pi (V))}, \qquad A \in \mathcal{F}. 
\]
Even if to construct $e^{\mathrm{i}t(H_0+\pi V)}$ we made use of the GNS representations, if the cutoff $g$ in the interaction Hamiltonian $V$ is kept of finite spatial support, the cocycle $U(t)$ can be constructed at an algebraic level. More precisely, $U(t)$ is obtained as a power series in the coupling constant whose coefficients are well defined elements of $\mathcal{F}$ (or of $\mathcal{F}_\beta)$. The convergence of this power series has been proven only in few examples. We shall discuss below how to obtain it in the considered framework.

\subsection{Interacting equilibrium state}

The interacting equilibrium state can be constructed a la Araki by 
\begin{equation}\label{eq:state-araki}
\omega^{\beta V}(A) = \frac{\omega^\beta(A U(\mathrm{i}\beta))}{\omega^\beta(U(\mathrm{i}\beta)) }, \qquad A \in \mathcal{F}\;.
\end{equation}
See \cite{Araki,OhyaPetz,FredenhagenLindnerKMS_2014,Galanda}. 
The expectation values in this state are given in terms of the connected or truncated  correlation functions
\begin{equation} \label{eq:state-araki-connected}
    \omega^{\beta V}(A) =
\omega^{\beta }(A)
+
\sum_{n\geq 1} (-1)^n
\int_{\beta {S}_n} \di^nu\;
\omega_C^\beta(A; \underbrace{\tau_{\mathrm{i}u_1}V;\dots \tau_{\mathrm{i}u_n}V}_{n})
\end{equation}
where $\omega^\beta_C$ denotes the connected correlation functions of the state $\omega^\beta$ and $\tau_{\mathrm{i}u}$ is the analytic continuation of the time evolution to the imaginary time. Moreover, $\beta {S}_n$ denotes  the $n$-dimensional simplex of edge $\beta$, defined as $\beta {S}_n = \{(u_1,\dots, u_n) | 0<u_1 <\dots <u_n < \beta\}$. 

The {\bf connected correlation functions} are recursively defined by the relations \cite{BKR78},
\begin{equation}\label{eq:recursive-relations-connected}
\omega^{\beta}(A_1\dots A_n)
=\sum_{\pi\in \mathcal{P}_n}\prod_{I\in \pi}\omega^\beta_C(A_{I_1};\dots;A_{I_{|I|}}), \qquad A_i\in \mathcal{F}
\end{equation}
where $\mathcal{P}_n$ is the set of partitions of $\{1,\dots, n\}$ and the elements of the subset $I\in \pi\in  \mathcal{P}_{n}$ are ordered in ascending order.
A direct expression for the connected correlation functions in $\mathcal{F}_\beta$, which we shall use below is obtained in the following:

\begin{proposition}\label{pr:connected-simple-graphs} The connected correlation functions admit the following expansion
\begin{equation}\label{eq:connnected-correlations-graph-expansion}
\omega_C^{\beta}(\alpha_\beta^{-1}A_1; \dots ;\alpha_\beta^{-1}A_n) = 
\sum_{G\in \mathcal{G}^{c}_n}\left. M \prod_{\{i,j\}\in E(G) \atop i<j}(e^{D_{ij}}-1)
A_1\otimes \dots \otimes  A_n \right|_{\Phi=0,\Phi^*=0},  \qquad A_i\in \mathcal{F}_\beta
\end{equation}
where $\mathcal{G}^c_n\subset \mathcal{G}_n$ is the set of simple connected graphs among $n$ vertices and the $i$-th and $j$-th functional derivative in $D_{ij}$ act respectively on $A_i$ and $A_j$.
\end{proposition}
\begin{proof}
Recalling equation \eqref{eq:eval-state-beta}, we have that 
\[
\omega^\beta(\alpha_\beta^{-1}A_1\dots \alpha_\beta^{-1}A_n)
=\left. A_1\star_\beta \dots \star_\beta A_n \right|_{\Phi=0,\Phi^*=0}.
\]
We have seen that from the definition of the $\star_\beta$ product given in equation \eqref{eq:star-product} and the definition of $D_{ij}$ given in equation \eqref{eq:Dij}, the decomposition in terms of a sum over connected graphs given in equation \eqref{eq:decomposition-connected} holds, hence
\begin{align}\label{eq:decomposition-connected-state}
\omega^{\beta}(\alpha_{\beta}^{-1}A_1 \dots \alpha_\beta^{-1}A_n) 
&= 
\sum_{\pi\in \mathcal{P}_n}\prod_{I\in \pi}
\left(
 M\sum_{G\in \mathcal{G}^c(I)}\left.  
 \prod_{\{i, j\}\in E(G) \atop i<j}(e^{D_{ij}}-1)
 A_{I_1}\otimes \dots \otimes  A_{I_{|I|}} \right)
\right|_{\Phi=0,\Phi^*=0}\!\!\!\!\!\!\!\!\!\!\!\!\!\!\!\!\!\!\!\!\!\!. 
\end{align}
Comparing equation \eqref{eq:decomposition-connected-state} with the recursive relations for the connected correlation functions given in equation \eqref{eq:recursive-relations-connected}, and recalling that the one-point functions coincide with the truncated one-point functions ($\omega^\beta(A_i) = \omega^\beta_C(A_i)$), we have the thesis.
\end{proof}


The action of the automorphisms $\tau_t$, on factors of the connected correlation functions at the right hand side of \eqref{eq:state-araki-connected}, are taken into account shifting the time at which the two-point functions are considered. For this reason, the analytic continuation can be taken thanks to the analytic property of the connected correlation functions which descends from the KMS condition.
In particular, the thermal propagator we are introducing below is used in place of the two-point functions $\omega_{2,\pm}^\beta$ in the expansion in terms of connected correlation functions in \eqref{eq:state-araki-connected}.
Notice furthermore that $\omega^\beta(U(\mathrm{i}\beta))$ denotes the  partition function of the state $\omega^{\beta V}$ and hence $U(\mathrm{i}\beta)$ can be understood as the relative partition function of $\omega^{\beta V}$ with respect to $\omega^\beta$.

\subsection{Description of the relative partition function $U(\mathrm{i}\beta)$}

Notice that $U(\mathrm{i}\beta)$ is the (anti) time ordered exponential with respect to the imaginary time of the interaction Hamiltonian. In the previous section, the imaginary time has been used as an external parameter of the theory. Here, in order to concretely present  $U(\mathrm{i}\beta)$ and to study its convergence, it is useful to use the imaginary time as a coordinate and to use fields which depend and are smeared also on this imaginary time. 

Considering the propagators and the correlation functions of fields depending on the imaginary time $u$, corresponds to consider a Euclidean quantum field theory. The original quantum field theory at finite temperature is then recovered in suitable limits at $u=0$ of the correlation functions of this euclidean theory. So, the field configurations we shall consider to describe the extended Euclidean quantum field theory are denoted by $\Phi$ and $\Phi^*$ and are complex valued functions on $\Sigma \times [0,\beta]$ (extended by periodicity to $\Sigma \times \mathbb{R}$).

The starting point is the following definition.
\begin{definition}
The algebra of the {\bf Euclidean field theory} $\mathcal{E}$
 is the commutative (polynomial) $*$-algebra generated by linear fields $\Phi(f)$, $\Phi^*(h)$ smeared with elements of $f,h\in L^2(\mathbb{R}^3\times[0,\beta])$ and by the normal ordered Wick squares $:\!\!|\Phi|\!\!:^2(g)$ smeared with $g\in C^\infty_0(\mathbb{R}^3\times[0,\beta])$.
\end{definition} 
The linear fields in $\mathcal{E}$ are thus formally described as
\begin{equation}\label{eq:euclidean-fields-integrals}
\Phi(f) = \int_{\Sigma \times [0,\beta]} \Phi(x,u) \overline{f(x,u)} \di^3 x \di u,
\qquad
\Phi^*(f) = \int_{\Sigma \times [0,\beta]} \Phi^*(x,u) f(x,u) \di^3 x \di u,
\end{equation}
where $f\in L^2 (\Sigma \times [0,\beta])$, while the Wick square as
\[
:\!\!|\Phi|\!\!:^2(g) = \int :\Phi\Phi^*:(x,u) g(x,u) \di^3 x \di u  
\]
where $g \in C^{\infty}_0(\Sigma \times [0,\beta])$. 
The product in this algebra $\mathcal{E}$ is commutative 
and the $*$-operation is also naturally implemented in $\mathcal{E}$.
 
%
%



It holds in particular that the fields of $\mathcal{F}$ given in \eqref{eq:linear-fields-integrals} 
%
evolved with the $*$-automorphisms describing the time evolution and analytically continued at imaginary time $u$, correspond to the fields of the Euclidean theory given in \eqref{eq:euclidean-fields-integrals}. Hence, at least formally we have
%
\begin{equation}\label{eq:tau-field}
\Phi(x,u) := \tau_{iu} \Phi(x), \qquad \Phi^*(x,u) := \tau_{iu} \Phi^*(x).
\end{equation}

The free Euclidean theory we are discussing has to be equipped with a covariance described by means of the 
thermal propagator of the theory.

\begin{definition}\label{de:thermal-propagator}
The {\bf thermal propagator} of the free theory associated with $K$ given in \eqref{eq:def-K} and denoted by $\Delta^\beta$ is the integral kernel of the operator valued periodic function of $u$ with period $\beta$ and which is 
\begin{equation}\label{eq:thermal-propagator}
\Delta^\beta(u)= \frac{e^{-u K}}{1-e^{-\beta K}}
 \qquad u\in [0,\beta), \qquad
 \Delta^\beta(u+\beta)=\Delta^\beta(u), \qquad u\in \mathbb{R}
\end{equation}
where 
its spatial Fourier transform (for $u\in [0,\beta)$) is 
\[
\hat{\Delta}^\beta(u,p) = \frac{e^{-u(\frac{p^2}{2 m}-\mu)}}{1-e^{-\beta(\frac{p^2}{2m}-\mu)}}. 
 \]
Notice that $\Delta^\beta(u)$ is regular for $u\neq 0$ and has a discontinuity in $u=0$. 
The {\bf covariance} in $\mathcal{E}$ associated to $\Delta^\beta(u)$ is described by the operator 
\begin{equation}\label{eq:covariance}
 \Delta^\beta \psi(x,u) = \int \di \overline{u} \Delta^\beta(\overline{u}-u)\psi(\cdot,\overline{u}))(x) 
 \qquad \psi \in L^{2}(\mathbb{R}^3\times \mathbb{R}).
\end{equation}
\end{definition}

The thermal propagator is used to describe the covariance of the Euclidean quantum field theory we are considering and, at the same time, to realize the time ordering with respect to the imaginary time necessary to describe $U(i\beta)$. 

It is also used to obtain a representation of $\mathcal{E}$ in terms of objects which do not contain divergences. Namely, similarly to the representation of $\mathcal{F}$ in terms $\mathcal{F}_\beta$, we denote by $\mathcal{E}_\beta$ the representation of $\mathcal{E}$ and we describe it as an algebra of functionals over suitable field configurations.
The field configurations we shall consider are denoted by  $\Phi$ and $\Phi^*=\overline{\Phi}$, complex valued functions on $\Sigma \times [0,\beta]$ (extended by periodicity to $\Sigma \times \mathbb{R}$). Then,
the map which implements the representation of $\mathcal{E}$ to $\mathcal{E}_\beta$ (similar to $\alpha_\beta:\mathcal{F}\to\mathcal{F}_\beta$ introduced in \eqref{eq:alphabeta}) is the map which also realises the (anti) {\bf  time ordering} with respect to the imaginary time
\begin{equation}\label{eq:Taubeta}
    \mathcal{T}_\beta F=e^{\gamma}F
\end{equation}
where 
\[
\gamma = \int \di^3 x \di^3 x'\di u\di u'  \Delta^\beta(x-x',u'-u) \frac{\delta^2}{\delta \Phi(x,u) \delta \Phi^*(x',u')}
\]
given in terms of the integral kernel of the covariance in \eqref{eq:covariance}, given in terms of the thermal propagator in \eqref{eq:thermal-propagator}.

This map acts on the polynomial algebra of smeared linear fields. It can be further extended to algebra generated by linear fields and the Wick square described above with the observation that  $|\Phi|^{2} = \Phi\Phi^* - c$ where $c$ does not depend on fields and implements the normal ordering. Then, the generators of $\mathcal{E}_\beta$ are now $\mathcal{T}_\beta \{\mathbb{1}, \Phi, \Phi^*, |\Phi|^2\}$,
and the map $\mathcal{T}_\beta$ intertwines the  
{\bf commutative product} on $\mathcal{E}$ with the 
commutative product in $\mathcal{E}_\beta$ which 
is here denoted by $\cdot_T$
\[
\mathcal{T}_\beta ( A B) = 
(\mathcal{T}_\beta A) \cdot_T
(\mathcal{T}_\beta B).
\]
Moreover, since $\mathcal{T}_\beta \Phi(f)=\Phi(f)$ it also holds that
\[
\Phi(f_1) \cdot_T ... \cdot_T\Phi(f_n)\cdot_T \Phi^*(g_1)\cdot_T...\cdot_T\Phi^*(g_n)
=
\mathcal{T}_\beta (\Phi(f_1)...\Phi(f_n)\Phi^*(g_1)...\Phi^*(g_n) ). 
\]
and the evaluation on the free equilibrium states $\omega^{\beta}$ of $F\in\mathcal{E}$ consists in testing $\mathcal{T}_\beta F \in \mathcal{E}_\beta$ on the vanishing field configurations
\[
\left. \mathcal{T}_\beta F \right|_{\Phi=0\Phi^*=0}.
\]

Let us discuss how the original theory is recovered. The two-point function of the Euclidean theory is
\[
\omega^{\beta}(\Phi(x,u)\Phi^*{(x',u'})) = \Delta^\beta(x-x',u'-u)
= \left. \mathcal{T}_\beta \Phi(x,u)\Phi^*{(x',u'})
\right|_{\Phi=\Phi^*=0}.
\]
Therefore, the imaginary time $u=0$ correlation functions (of $\omega^\beta$ on $\mathcal{F}$) are obtained as suitable limits of the correlation functions of the Euclidean theory. In particular, it holds that 
\[
\lim_{u\to 0^\pm} \Delta^\beta(u) = B_\mp.
\]
Hence, for the two-point function of the theory, 
\[
\omega^{\beta}(\Phi(x)\Phi^*{(x'})) = 
\lim_{u'-u\to 0^+}\Delta^\beta(x-x',u'-u),
\qquad
\omega^{\beta}(\Phi^*(x)\Phi{(x'})) = 
\lim_{u'-u\to 0^-}\Delta^\beta(x-x',u'-u)
\]
and the commutation relations of $\mathcal{F}$ are reconstructed as
\[
\begin{aligned}
\omega^{\beta}(\Phi(x)\Phi^*{(x'}))
-
\omega^{\beta}(\Phi^*(x')\Phi(x))
&= 
\lim_{u'-u\to 0^+}\Delta^\beta(x-x',u'-u)
-
\lim_{u'-u\to 0^-}\Delta^\beta(x-x',u'-u)
\\
&= 
\lim_{\epsilon \to 0^+}
\left(\Delta^\beta(x-x',\epsilon)
-
\Delta^\beta(x-x',\beta-\epsilon)
\right).
\end{aligned}
\]
In the case where there is no condensate, the functional whose expectation value in the state $\omega^\beta$ describes the relative partition function can now be written in $\mathcal{E}_\beta$ in terms of 
the interaction Hamiltonian $\tau_{\mathrm{i}u}V$ given in terms of $V$ in \eqref{eq:interaction-hamiltonian} and of 
\eqref{eq:tau-field}. It is thus seen as an element of $\mathcal{E}$, where the fields are smeared with functions on $\mathbb{R}^3\times [0,\beta]$. Namely, we have that
\[
U(\mathrm{i}\beta) = \mathcal{T}_\beta e^{-\int_0^\beta \di u \, \tau_{\mathrm{i}u}V}
\]
and, up to the convergence of the exponential which will be discussed below for the expectation values, we obtain an element of $\mathcal{E}_\beta$.
In particular, as a consequence of the above discussion, we have
\[
\omega^\beta(U(\mathrm{i}\beta))  = \left. \mathcal{T}_\beta e^{-\int_0^\beta \di u \tau_{\mathrm{i}u} V} \right|_{\Phi=0,\Phi^*=0}.
\]
Finally, the correlation functions of the interacting equilibrium state are obtained as suitable limits towards $u=0$. Namely, 
\[
\omega^\beta (\Phi(x) \Phi^*(x') U(\mathrm{i}\beta)) = \lim_{\epsilon \to 0^+} \left. 
\mathcal{T}_\beta
\left(
\Phi(x,u) \Phi^*(x',u+\epsilon) U(\mathrm{i}\beta) 
\right)
\right|_{\Phi=\Phi^*=0}
\]
and thus in the Euclidean theory there are all the elements to estimate the correlation function of the interacting equilibrium state. In particular, for the two-point function, we have
\[
\omega^{\beta V} (\Phi(x)\Phi^*(x')) 
=\frac{\omega^{\beta}(\Phi(x)\Phi^*(x') U(\mathrm{i}\beta))}{\omega^{\beta}( U(\mathrm{i}\beta))}.
\]










\section{Auxiliary theory: fields interacting with an external Gaussian field}\label{se:3}

In this section we discuss in which way the Gaussian auxiliary field (HS transformation) is introduced in order to split the study of convergence of the perturbative series defining the partition function in two steps. 

\subsection{Auxiliary field for vanishing background field}

We introduce an auxiliary classical field $A$ which is used to reconstruct a posteriori $U(i\beta)$. This field is a classical (commuting) external Gaussian stochastic potential periodic in the imaginary time of period $\beta$, of vanishing mean and characterized by a covariance expressed by the two-point function with integral kernel $\delta(u-u')v(x-x')$.
(here $\delta$ is the periodic delta function of period $\beta$).

We start discussing the case where no condensation is present. More precisely, instead of using directly $V$ in \eqref{eq:interaction-hamiltonian} to write $U(i\beta)$ and thus $\omega^{\beta V}$, we consider the following element of $\mathcal{E}$
\[
\mathcal{V}_A:=\int_0^{\beta} \di u  \int \di^3x :\!\!|\Phi|\!\!:^2(x,u)A (x,u) g(x).
\]
to construct the relative partition function. Namely, we consider it at the place of $\int \di u \tau_{\mathrm{i}u}V$ in the imaginary time ordered exponential, constructed with $\mathcal{T}_\beta$ and given in \eqref{eq:Taubeta}, with fields $\Phi(x,u) = \tau_{\mathrm{i} u} \Phi(x)$ and squared fields $:\!\!|\Phi|\!\!:^2(x,u) = \tau_{\mathrm{i} u} :\!\!|\Phi|\!\!:^2(x)$. Therefore, we study
\[
 U_A(\mathrm{i}\beta) = \mathcal{T}_\beta \exp \left( -(\mathcal{V}_A)\right),
\]
where $U_A(\mathrm{i}\beta)$ is a power series in the coupling constant whose coefficients are well defined elements of $\mathcal{E}_\beta$. We furthermore have that
\[
\omega^\beta (U_A(\mathrm{i}\beta)) = \left.  \mathcal{T}_\beta \exp \left( -(\mathcal{V}_A)\right) \right|_{\Phi=0,\Phi^*=0}.
\]
Starting from this expression, the relative partition function $\omega^{\beta}(U(\mathrm{i}\beta))$ and the correlation functions of the corresponding state, are later reconstructed adding a posteriori the internal lines corresponding to $v$ at the place of all possible couples $AA$. This step is realised evaluating the Gaussian field $A$ on the Gaussian state with zero mean and covariance $\delta(u-u')v(x-x')$, by applying the operation $e^{\Gamma}$, where 
\[
\Gamma = -\frac{1}{2} \int \di^3 x \di^3 x' \di u \di u'     v(x-x')\delta(u-u') \frac{\delta^2}{\delta A( x',u')\delta A( x,u)}
\]
and then by setting $A=0$. The consistency of this procedure is proven in the following proposition.

\begin{proposition}\label{prop:auxiliary-field}
Consider the power series in the coupling constant of the theory with coefficients in $\mathcal{E}_\beta$:
\[
U(\mathrm{i}\beta):= \mathcal{T}_\beta \exp (-\int_0^\beta    \tau_{\mathrm{i}u}{V} \di u )
\]
and
\[
U_A(\mathrm{i}\beta) :=  \mathcal{T}_\beta \exp (- {\mathcal{V}}_{A}).
\]
It holds that 
\begin{align*}
U(\mathrm{i}\beta)&= e^{\Gamma} \left. U_A(\mathrm{i}\beta)\right|_{A=0}
\\
&=  \left. e^{\Gamma} \mathcal{T}_\beta e^{- \mathcal{V}_{A}}\right|_{A=0}.
\end{align*}
\end{proposition}


\begin{proof}
Notice that $e^{\Gamma}$ and $\mathcal{T}_\beta=e^\gamma$, given in \eqref{eq:Taubeta}, act on independent fields. Respectively on $A$ and on $(\Phi,\Phi^*)$. Hence, they commute. More precisely, considering ${V}$ and ${\mathcal{V}}_A$ given by hypothesis, recalling that 
\[
\tau_{\mathrm{i}u} V = \frac{1}{2}\int :\Phi(x,u)\Phi^*(x,u) \Phi(y,u)\Phi^*(y,u): v(x-y) g(x) g(y) \di^3 x\di^3 y,
\]
we have that
\begin{align*}
e^{\Gamma}\left.
\mathcal{T}_\beta e^{-{\mathcal{V}}_A  }
\right|_{A=0}
&=
e^{\Gamma}\left.
e^{\gamma} e^{-{\mathcal{V}}_A  }
\right|_{A=0}
\\
&=
e^{\gamma}
e^{\Gamma}\left.
e^{-{\mathcal{V}}_A  }
\right|_{A=0}
\\
&=
e^{\gamma}
e^{-\int_0^\beta  \tau_{\mathrm{i}u}{V} \di u }.
\end{align*}
\end{proof}

We observe that $\mathcal{V}_A$ is quadratic in the field $\Phi$, for every choice of external field $A$. Furthermore, $\mathcal{V}_A$ is linear in $A$.

Following this strategy, in the next section, we construct the equilibrium state of the interacting theory introducing the auxiliary field $A$, evaluating the fields $\Phi$ and $\Phi^*$ in the corresponding state, and applying the evaluation of $A$ a posteriori. In order to implement this procedure, we observe that the interaction Hamiltonian $\mathcal{V}_A$ is quadratic in the fields $\Phi$ and for that field theory $A$ plays the role of an external potential. Hence, taking into account its influence on the state $\omega^\beta$ can be done exactly and not perturbatively, and the equilibrium state for the theory with the external potential $A$ can be studied exactly. The last step to discuss, is the prove that $e^{\Gamma}$ can be applied leading to a convergent series. 

\subsection{The case with non vanishing background, the condensate}

In the case with the condensate, the construction needs to be slightly modified in order to take into account a non vanishing background over which the bosonic quantum field theory is developed. We shall see that ideas similar of those discussed in the previous section are used. 

\subsubsection{Preliminary analysis}

Following the discussion given in section \ref{se:interacting-time-evolution}, the abstract Hamiltonian of the system is described by an Hamiltonian density given in terms of the field $\Phi$ as
%
%
%
\begin{align*}
\tilde{H} 
&=   \Phi^* (\mathcal{K}-\tilde\mu) \Phi +\frac{1}{2} |\Phi|^2 v*|\Phi|^2 
\end{align*}
where $\tilde\mu$ is the {\bf renormalized chemical potential} and, in the case where no adiabatic cutoffs are considered, it takes the form
\begin{equation}\label{eq:ren-mu}
\tilde\mu = \mu + v(0)
\end{equation}
see \eqref{eq:interaction-hamiltonian}.
We analyze how $\tilde{H}$
changes under a field redefinition
\[
\Phi = \phi_0 + \Psi, \qquad \Phi^*=\phi_0+ \Psi^*
\]
where $\phi_0$ is a classical real function which is the {\bf background value} of the field, representing the presence of a condensate, over which perturbation theory is analyzed. Correspondingly, $\Psi$ describes the fluctuations around the background $\phi_0$. We decompose the formal Hamiltonian density in contributions homogeneous in the fluctuations as follows
\[
\tilde{H}=\tilde{H}_0+\tilde{H}_1+\tilde{H}_2+\tilde{H}_3+\tilde{H}_4
\]
where
\begin{align*}
\tilde{H}_0&= \phi_0(\mathcal{K}-\tilde\mu)\phi_0 + \frac{1}{2}\phi_0^2 v *(\phi_0^2 )    
\\
\tilde{H}_1&= \phi_0(\mathcal{K}-\tilde\mu)(\Psi+\Psi^*) + \phi_0^2 v *(\phi_0(\Psi+\Psi^*))
\\
\tilde{H}_2&= \Psi^* (\mathcal{K}-\tilde\mu) \Psi + |\Psi|^2 v*(\phi_0^2) + \frac{1}{2}  \phi_0 (\Psi+\Psi^*) v* (\phi_0(\Psi+\Psi^*))
\\
\tilde{H}_3&= |\Psi|^2 v* ((\Psi+\Psi^*)\phi_0)
\\
\tilde{H}_4&= \frac{1}{2} |\Psi|^2 v* |\Psi|^2.
\end{align*}
At this stage, we choose the classical function $\phi_0$ as a local stationary minimiser of $\tilde{H}_0$. With this choice, $\tilde{H}_1$ is a total derivative on $\Sigma$ and hence $\tilde{H}_1$ vanishes once integrated on $\Sigma$. We shall thus choose constant $\phi_0$ which satisfies
\[
\|v\|_1\phi^2_0 =\tilde\mu ,
\]
and hence 
\begin{equation}\label{eq:phi0}
\phi^2_0 =\frac{\tilde\mu}{\|v\|_1}. 
\end{equation}
Moreover, in order to shift some quadratic contributions of the interaction Hamiltonian into the free Hamiltonian, modifying the thermal propagators accordingly, we will make use of the perturbative agreement \cite{HW05, DragoHackPinamonti}. 

With this in mind, we consider the following splitting into free Hamiltonian density
\begin{equation}\label{eq:free-hamiltonian-condensate}
H_{20} = \Psi^* (\mathcal{K}- \tilde{\tilde{\mu}})\Psi,
\end{equation}
where $\tilde{\tilde{\mu}} = \tilde{\mu}-\phi_0^2 \|v\|_1$ is the renormalized chemical potential, and interaction Hamiltonian density
\begin{equation}\label{eq:interaction-Hamiltonian-condensate}
{H}_I =  \frac{1}{2}  \phi_0 (\Psi+\Psi^*) v* (\phi_0(\Psi+\Psi^*)) 
+ 
|\Psi|^2 v* ((\Psi+\Psi^*)\phi_0)
+ 
\frac{1}{2} |\Psi|^2 v* (|\Psi|^2)
\end{equation}



\subsection{Auxiliary field for non vanishing background quantum fields}

In this section we shall introduce the auxiliary external field $A$, which will be eventually removed with the application of $e^{\Gamma}$.
%
We start considering the following redefined Hamiltonian density decomposed in homogeneous contributions in $\Psi$ and $\Psi^*$
\[
\tilde{\tilde{H}} = 
\tilde{\tilde{H}}_0 +
\tilde{\tilde{H}}_1
+\tilde{\tilde{H}}_2
\]
where
\begin{align*}
\tilde{\tilde{H}}_0&=
\tilde{H}_0
\\
\tilde{\tilde{H}}_1&= \tilde{H}_1+
g \phi_0 A(\Psi+\Psi^*) = g  A\phi_0(\Psi+\Psi^*)
\\
\tilde{\tilde{H}}_2&= \Psi^* (\mathcal{K}-\tilde{\mu}) \Psi + |\Psi|^2v*(\phi_0^2)+ g A |\Psi|^2. 
\end{align*}
where $\tilde{H}_1$ is a total derivative and thus, it can be discarded in the analysis of the relative partition function. Furthermore, we added an {\bf adiabatic cutoff} described by a smooth function $g \geq 0$ on $\Sigma$ which has compact support, it is equal to $1$ on large regions of space and which bounds the spatial domain where the interaction takes place.

Notice that if $\phi_0$ is constant, the chemical potential is renormalized in the quadratic theory
\[
\tilde{\tilde{H}}_2 = \Psi^* (\mathcal{K}-\tilde{\tilde\mu}) \Psi + gA|\Psi|^2 
\]
where now we are considering
\begin{equation}\label{eq:renormalizaed-chemical-potential}
\tilde{\tilde\mu} = \tilde{\mu}- \phi_0^2\|v\|_1 - \epsilon= -\epsilon.
\end{equation}
Here, we have used the fact that with $\tilde{\mu}$ given in \eqref{eq:ren-mu}, and with the choice of $\phi_0$ made in \eqref{eq:phi0}, $\tilde{\tilde{\mu}} = -\epsilon$ with $\epsilon>0$ an artificial {\bf regularization parameter} used to keep the theory finite. In this way, the renormalized chemical potential $\tilde{\tilde{\mu}}=-\epsilon$ is finite and strictly negative. We shall take the limit $\tilde{\tilde{\mu}} \to 0$, ($\epsilon \to 0$) only a posteriori, on a general level, together with the limit $g\to1$, the latter only for certain regimes. 

Notice that the Hamiltonian density of the free theory is $H_{20} = \Psi^* (\mathcal{K}- \tilde{\tilde{\mu}}) \Psi$, which coincides with $H_{20}$ given in equation \eqref{eq:free-hamiltonian-condensate} up to the $\epsilon$ regularization introduced in this step, used to make the spectrum of the free Hamiltonian bounded away from $0$.

With this in mind, we discuss now how to construct the relative partition function of this interacting auxiliary theory denoted by $U_A^{\phi_0}(\mathrm{i}\beta)$. The contribution which contains $\tilde{\tilde{H}}_0$ is a constant factor and thus it can be discarded in the analysis of the relative partition function.
%
%
The interaction Hamiltonian of the auxiliary theory which is then used to get $U_A^{\phi_0}(\mathrm{i}\beta)$ is related to the contributions in $\tilde{\tilde{H}}$ which are linear in $A$ and it has the form
\[
V_A = \int_0^\beta Q_A(u) \di u.
\]
where 
\begin{equation}\label{eq:QA}
Q_A(u) = \int_\Sigma \di^3 x  g(x) A(x,u) (|\Psi|^2(x,u) + (\Psi(x,u)+\Psi^*(x,u))\phi_0 ).
\end{equation}
Notice that $V_A$ is a well defined element of $\mathcal{E}$ and furthermore it is at most quadratic in the field $\Psi$. $Q_A(u)$ is thus understood as a distribution over $C^\infty([0,\beta])$ with values in $\mathcal{E}$.
%
The following proposition discusses the connection of the relative partition function of the auxiliary theory with that corresponding to $\tau_{\mathrm{i}u }V$.

\begin{proposition}\label{prop:auxiliary-partition-function}
Considering the equilibrium state $\omega^\beta$ associated to the free time evolution generated by $H_{20}$, given in \eqref{eq:free-hamiltonian-condensate} in terms of $\mathcal{K}-\tilde{\tilde{\mu}}$ where $\tilde{\tilde{\mu}}$ is the one given in \eqref{eq:renormalizaed-chemical-potential}, it holds that
\[
 U(\mathrm{i}\beta) = \mathcal{T}_\beta e^{- \int_0^\beta \di u \int \di^3 x H_I}
\]
in the sense of power series in the coupling constant  with coefficients in $\mathcal{E}_\beta$, where the interaction Hamiltonian is
\begin{equation}\label{eq:HI}
H_I =  \frac{1}{2}  g\phi_0 (\Psi+\Psi^*) v* (g\phi_0(\Psi+\Psi^*)) 
+ 
g|\Psi|^2 v* ((\Psi+\Psi^*)g\phi_0)
+ 
\frac{1}{2} g|\Psi|^2 v* (g|\Psi|^2)
\end{equation}
as given in \eqref{eq:interaction-Hamiltonian-condensate}.
The associated relative partition function is given by the expectation value of
\[
U(\mathrm{i} \beta) = \left. e^{\Gamma} U_A^{\phi_0}(\mathrm{i}\beta) \right|_{A=0},
\]
where 
\[
U_A^{\phi_0}(\mathrm{i}\beta) =
\mathcal{T}_\beta e^{-V_A},
\qquad \text{with}\qquad
V_A =\int_0^\beta \di u \int \di^3 x\;  g(x) A(x,u)  (|\Psi|^2+(\Psi+\Psi^*)\phi_0) 
\]
and defined in the sense of formal power series with coefficients in $\mathcal{E}_\beta$
\end{proposition}
\begin{proof}
The result is similar to proposition \ref{prop:auxiliary-field}. In particular the application of $e^{\Gamma}$ and of $\mathcal{T}_\beta$ commute because they act on independent fields. 
Therefore, recalling the form of $H_I$ in \eqref{prop:auxiliary-partition-function} 
and the form of $Q_A$ given in \eqref{eq:QA}, we get that 
\[
e^{\Gamma} \mathcal{T}_\beta e^{-\int_0^\beta \di u\, Q_A(u)} = 
e^{\Gamma} \mathcal{T}_\beta e^{-\int_0^\beta \di u\, Q_A(u)}
=
\mathcal{T}_\beta e^{-  \int_0^\beta \di u \int \di^3 x H_I }
\]
thus getting the thesis.
\end{proof}

\begin{rem}
In the theory presented so far, we are considering a non vanishing background $\phi_0$ which is constant on the whole space. To get a well defined interaction Hamiltonian and control its exponential, we have inserted an adiabatic cutoff $g$ in $H_I$. In this way, $\int du \int \di^3 x H_I$ is a well defined element of $\mathcal{E}$ and its time ordered exponential can be considered. We might eventually remove this cutoff taking the limits where $g\to 1$ everywhere on $\Sigma=\mathbb{R}^3$ in a suitable way. The role of the regularization parameter $\epsilon = - \tilde{\tilde{\mu}}$ in the renormalized chemical potential is similar. It is used to make the free theory well behaved in the infrared regime. The limit where $\epsilon\to0$ is eventually taken after having proven the convergence of the series defining the Gibbs states at finite temperature. 
\end{rem}

Therefore, in view of Proposition \ref{prop:auxiliary-partition-function}, we shall furthermore use an auxiliary state for the intermediate theory to built $\omega^{\beta V}$, constructed as
\begin{equation}\label{eq:omegabetaA}
\omega^{\beta A}(B) := \frac{\omega^{\beta}(BU^{\phi_0}_A(\mathrm{i}\beta))}{\omega^{\beta}(U^{\phi_0}_A(\mathrm{i}\beta))}
=\sum_{{n}\geq 0} (-1)^n\int_{\beta S_n} \di^n u\;\omega_C^\beta(B;\underbrace{Q_A(u_1);\dots; Q_A(u_n)}_n) , \qquad B\in \mathcal{E}.
\end{equation}
We furthermore observe that, since the $V_A$ and $Q_A$ are quadratic in the fields, the two-point function of this state can be directly constructed and is quasi-free. Here, we are implicitly applying the principle of perturbative agreement \cite{HW05, DragoHackPinamonti}. 


\subsection{Intermediate theory and evaluation of the auxiliary field $A$}

We consider the case $\phi_0\neq 0$ and analyze the relative partition function 
\[
Z=  e^{W} = \omega^\beta( U(\mathrm{i}\beta))
\]
and among all possible correlation functions of the original model, without loss of generality, the two-point function of the interacting theory
\[
\langle f, Sh \rangle :=\omega^{\beta V} (\Psi(f)\Psi^*(h)) = \frac{\omega^{\beta}
(\Psi(f)\Psi^*(h) U(\mathrm{i}\beta))}
{\omega^\beta( U(\mathrm{i}\beta))} .
\]
The following proposition connects the expectation values of the interacting theory in the state $\omega^{\beta V}$ to expectation values of the theory with the external auxiliary field $A$ in the state $\omega^{\beta A}$ given in equation \eqref{eq:omegabetaA} for the case with non vanishing background $\phi_0$. 

\begin{proposition}\label{pr:what-we-compute}
Denote by 
\[
\langle f,  \Omega^{\phi_0}  g  \rangle = 
\omega^{\beta A} (\Psi(f)\Psi^*(g)),
\]
and by 
\[
Z_A=e^{W_A} = \omega^{\beta } (U^{\phi_0}_A(\mathrm{i}\beta)).
\]
The evaluation at $\Psi,\Psi^*=0$ before the application of $e^{\Gamma}$ can be considered and it holds that
\begin{equation}\label{eq:S}
S = \left.\frac{e^{\Gamma}\left(\Omega^{\phi_0} e^{W_A}
\right)}{e^\Gamma (e^{W_A})}\right|_{A=0},
\end{equation}
as well as 
\begin{equation}\label{eq:partition-fuction}
Z = \left.e^{\Gamma} e^
{W_A}\right|_{A=0}.
\end{equation}
\end{proposition}
\begin{proof}
To prove these claims, we observe that for any $F = F(\Psi,\Psi^*)$:
\[
\omega^{\beta} (F) = \left. F   \right|_{\Psi=0,\Psi^*=0}.
\]
Hence, considering a generic functional $G = G(A,\Psi,\Psi^*)$:
\begin{align*}
\omega^{\beta} (e^{\Gamma} G) &=\left. (e^{\Gamma} G) \right|_{\Psi=0,\Psi^*=0}
\\
&=\left. e^{\Gamma} (G \right|_{\Psi=0,\Psi^*=0}) = e^{\Gamma} \omega^{\beta}(G),
\end{align*}
because $A$ and $\Psi$ are independent fields and the corresponding functional derivatives in $\Gamma$ and in $\mathcal{T}_\beta$ commute. With this observation, and recalling the result of Proposition \ref{prop:auxiliary-partition-function}, we notice that 
\[
Z=\omega^{\beta}(U(i\beta)) 
=
\left. \omega^{\beta}(e^{\Gamma} U^{\phi_0}_A(\mathrm{i}\beta))
\right|_{A=0}
= \left. e^{\Gamma}\omega^{\beta}(U^{\phi_0}_A(\mathrm{i}\beta))
\right|_{A=0} = \left. e^{\Gamma}Z_A
\right|_{A=0}.
\]
Furthermore,
\[
\omega^{\beta V}(\Psi(f)\Psi^*(g))
=\frac{\omega^{\beta}(\Psi(f)\Psi^*(g) U(\mathrm{i}\beta))}{\omega^{\beta }(U(\mathrm{i}\beta))}
=\left. \frac{\omega^{\beta}(\Psi(f)\Psi^*(g) e^{\Gamma} U^{\phi_0}_A(\mathrm{i}\beta))}{\omega^{\beta }(e^{\Gamma} U^{\phi_0}_A(\mathrm{i}\beta))}
\right|_{A=0}
=\left.\frac{e^{\Gamma}\omega^{\beta}(\Psi(f)\Psi^*(g) U^{\phi_0}_A(\mathrm{i}\beta))}{e^{\Gamma}\omega^{\beta }(U^{\phi_0}_A(\mathrm{i}\beta))}\right|_{A=0}.
\]
Then, recalling that
\[
\omega^{\beta A} (\Psi(f)\Psi^*(g)) = \frac{\omega^{\beta} (\Psi(f)\Psi^*(g)U^{\phi_0}_A(\mathrm{i}\beta))}{\omega^{\beta} (U^{\phi_0}_A(\mathrm{i}\beta))}
\]
we obtain that 
\[
\omega^{\beta} (\Psi(f)\Psi^*(g)U^{\phi_0}_A(\mathrm{i}\beta))=
\omega^{\beta A}(\Psi(f)\Psi^*(g)) \omega^{\beta} (U^{\phi_0}_A(\mathrm{i}\beta)).
\]
Combining these two observations we have 
\[
\omega^{\beta V}(\Psi(f)\Psi^*(g))
=
\left.\frac{e^{\Gamma}\left(\omega^{\beta A}(\Psi(f)\Psi^*(g))  \omega^{\beta}(U^{\phi_0}_A(\mathrm{i}\beta))\right)}{e^{\Gamma}\omega^{\beta }(U^{\phi_0}_A(\mathrm{i}\beta))}
\right|_{A=0}
=
\left.\frac{e^{\Gamma}\left(\langle f,\Omega^{\phi_0} g \rangle Z_A\right)}{e^{\Gamma}Z_A}\right|_{A=0},
\]
from which the thesis follows.
\end{proof}

\begin{rem}\label{rem: 2puntinpunti}
The result of this proposition implies that an hypothetical convergence established for the two-point function, is extended to all $n$-point functions of the state $\omega^{\beta V}$. The key idea is to exploit the quasifree property of the equilibrium state of the auxiliary theory $\omega^{\beta A}$, along with the convergence estimates already established for its one-point and two-point functions.
\end{rem}

We furthermore observe that because of the evaluation
of the field $A$ in a gaussian state of vanishing mean, namely setting $A=0$ after applying $e^{\Gamma}$,
the contribution with an odd number of $A$ vanish both 
in \eqref{eq:S} and in \eqref{eq:partition-fuction}. 


We conclude this section recalling that the following expansions in terms of connected correlations hold
\begin{equation}
\label{eq:logZA}
W_A=\log Z_A = \log(\omega^{\beta} (U_A^{\phi_0})) 
=
\sum_{n\geq 0} (-1)^n\int_{\beta S_n}\di^n u\, \omega_C^\beta(Q_A(u_1);\dots; Q_A(u_n)).
\end{equation}
where ${V}_A = \int_0^\beta du Q_A \in \mathcal{E}$ and
\begin{equation}\label{eq:evomegaA}
\omega^{\beta A}(F) 
= \sum_{n\geq 0} (-1)^n 
\int_{\beta S_n} \di^n u \, \omega_C^\beta (F;Q_A(u_1);\dots ;Q_A(u_n) ).
\end{equation}

\subsection{Relative entropy}
We notice in this section that 
\[
W_A=\log Z_A = \log \omega^{\beta}(U_A^{\phi_0}(\mathrm{i}\beta))
\]
has a direct interpretation in terms of the relative entropy of $\omega^{\beta A}$ relative to $\omega^{\beta}$. 
Actually, see e.g. the book of Ohya and Petz \cite{OhyaPetz} but also \cite{DrFaPi} for the analysis of that formula in the context of field field theories, we have
\begin{equation}\label{eq:relative-entropy-basic}
\mathcal{S}_{\phi_0}:=\mathcal{S}(\omega^\beta,\omega^{\beta A}) =  \omega^\beta( {V}_A)  + \log(\omega^\beta(U^{\phi_0}_{A}(\mathrm{i}\beta))).
\end{equation}
In particular, by the linked cluster theorem, $\omega^{\beta}(U^{\phi_0}_{A}(\mathrm{i}\beta))$ can be written as the exponential of the sums of connected components (sum of connected graphs obtained by means of the imaginary time ordered exponential of equation \eqref{eq:logZA}). Moreover, since the interaction Hamiltonian is particularly simple, there are three types of graphs which are summed. The contributions corresponding to $-\omega^{\beta}({V}_A)=T_0$. The contributions corresponding to lines with $\phi_0A(\Psi+\Psi^*)$ at the extreme points and with $A|\Psi|^2$ as internal points. This correspond to $T_1$. The contributions $T_2$ formed by the sum of all possible loops with at least two vertices and with $A|\Psi|^2$ as internal points. Hence, 
\[
\omega^{\beta}(U^{\phi_0}_{A}(\mathrm{i}\beta)) =
e^{T_0+T_1+T_2}
\]
and therefore in terms of the relative entropy
\[
T_2 = \left. \mathcal{S}(\omega^\beta,\omega^{\beta A})\right|_{\phi_0=0} = 
\mathcal{S}_0(\omega^\beta,\omega^{\beta A}),
\qquad 
T_1= \mathcal{S}_{\phi_0}-
\mathcal{S}_0,
\qquad 
T_0=-\omega^{\beta}({V}_A),
\]
where $\mathcal{S}_{\phi_0}  = \mathcal{S}(\omega^\beta,\omega^{\beta A})$ for $\phi_0\neq 0$.
We shall analyze the precise form of these contributions in terms of the thermal propagators below in order to analyze how they precisely depend on $A$.

\section{Quadratic interaction}\label{se:4}


In this section we discuss the role of the quadratic Hamiltonian ${V}_A$ and obtain the explicit form of the correlation function of $\omega^{\beta A}$ introduced in \eqref{eq:omegabetaA}. This is a necessary prerequisite 
to study the form of the expectation values of observables of the original interacting theory.

\subsection{Two-point function of the state $\omega^{\beta A}$}

We discuss how the correlation functions associated to the state
$\omega^{\beta A}$ correspond to correlation function of the linear theory described by the Hamiltonian $K+A$
and how they
are related to those of the free theory where the Hamiltonian is $K$. In this section $A$ is a compactly supported smooth external potential which depends also on the imaginary time $u$. Later we shall substitute $gA$ at the place of $A$ and the Hamiltonian $K=\mathcal{K}-\tilde{\tilde{\mu}}$.\\\\



Let us denote by 
\[\Omega_0 = \frac{1}{1-e^{-\beta K}}
\] 
the operator on the one particle Hilbert space $L^2(\mathbb{R}^3;\mathbb{C})$ which realises the (truncated) two-point function of the free theory
\[
\langle f,\Omega_0 h \rangle = \omega^{\beta }(\Psi(f)\Psi^*(h)) ,
\]
and by $\Omega$ the operator realising the (truncated) two-point function of the perturbed theory
\[
\langle f,\Omega h \rangle
=
\langle f,\Omega^{\phi_0} h \rangle
-
\omega^{\beta A}(\Psi(f))
\omega^{\beta A}(\Psi^*(h))
=\omega^{\beta A}(\Psi(f)\Psi^*(h))
-
\omega^{\beta A}(\Psi(f))
\omega^{\beta A}(\Psi^*(h))
\]
where $\Omega^{\phi_0}$ is the operator realising the two-point function introduced in Proposition \ref{pr:what-we-compute}. We shall prove that they correspond to having an Hamiltonian $K+A$ on the one particle Hilbert space where $A$ acts as an external potential.

To get this claim, we use ordinary perturbation theory to connect $\Omega$ with $\Omega_0$ considering $A$ as an external potential. The state $\omega^{\beta A}$ is obtained as in \eqref{eq:state-araki} and in particular with the formula \eqref{eq:state-araki-connected} with the perturbation Lagrangian density $\Phi^*\Phi A$. More precisely, \eqref{eq:state-araki-connected} results in an ordinary Dyson series given in terms of the thermal propagator of the free theory and the perturbation external potential $A$.

\begin{proposition}\label{prop:one-point-two-point}
For a compactly supported function $A$, consider $Q_A$ given in \eqref{eq:QA} and its decomposition in contributions homogeneous in the number of fields $\Psi$
    \[
    Q_A(u) = 
    Q_{A,1}(u)+Q_{A,2}(u)\;
\qquad
    Q_{A,1}=    \int_\Sigma \di^3 x  A(u,x) (\Psi+\Psi^*)\phi_0 
\qquad
    Q_{A,2}=    \int_\Sigma \di^3 x   A(u,x) |\Psi|^2 .
    \]
Consider the connected correlation functions $\omega^\beta_C$ of the state $\omega^\beta$ given in \eqref{eq:omegabeta} and the following functions 
\[
(u_1,\dots,u_n)\mapsto \omega^{\beta}_C(\tau_{\mathrm{i}u_1}F_1(u_1);\dots  ; \tau_{\mathrm{i}u_n}F_n(u_n)))
\]
defined by the analytic property of the correlation function of the state $\omega^{\beta}$ on $\beta S_n$, and extended by symmetry to $(0,\beta)^n$ for various function $F_i(u)$ with values in the observable algebra $\mathcal{F}$ (recall that $\int du {\tau}_{\mathrm{i}u} F(u)$ is in $\mathcal{E}$). 
    
We have that 
\[
\omega^{\beta A}(\Psi(f)) = -\sum_{n\geq 0}\frac{(-1)^n}{n!}
    \int_0^\beta \di u_0
    \int_{B_n} \di^n u\,
    \omega_C^\beta( \Psi(f); Q_{A,2}(u_1);\dots; Q_{A,2}(u_n);Q_{A,1}({u_0}))
\]
\[
\omega^{\beta A}(\Psi^{*}(h)) = -\sum_{n\geq 0}\frac{(-1)^n}{n!}
    \int_0^\beta \di u_0
    \int_{\beta B_n} \di^n u\,
    \omega_C^\beta( \Psi^*(h); Q_{A,2}(u_1);\dots; Q_{A,2}(u_n);Q_{A,1}({u_0}))
\]
\[
\langle f, \Omega h\rangle
=
\sum_{n\geq 0}\frac{(-1)^n}{n!}
    \int_{\beta B_n} \di^n u\,
    \omega_C^\beta( \Psi(f)\Psi^*(h); Q_{A,2}(u_1);\dots; Q_{A,2}(u_n))
\]
where $\omega^\beta_C{}$ are the connected correlation function of the state \eqref{eq:omegabeta} and $\beta B_n=(0,\beta)^n$.
In particular, it holds that 
the linear contribution in $Q_A$ has no influence in $\Omega$.
Furthermore,  
\begin{equation}\label{eq:truncated-two-point}
\langle f, \Omega h\rangle = 
\langle f,\Omega_0 h\rangle
+\sum_{n\geq 1} (-1)^n \int_{(0,\beta)^n} \di^n u\,\langle f, \Delta^\beta(u_1)A(u_1)\Delta^\beta(u_2-u_1)\dots \Delta^\beta(u_{n}-u_{n-1}) A(u_n)\Delta^\beta(-u_n) h\rangle 
\end{equation}
where $\Delta^\beta$ is the thermal propagator given in \eqref{eq:thermal-propagator} and $(A(u)f)(x)=A(x,u)f(x)$ for any $f$ in the one particle Hilbert space.
\end{proposition}
\begin{proof}
The expectation values in the state $\omega^{\beta A}$ can be given in terms of connected correlation functions as in equation \eqref{eq:evomegaA} where integrals over the $n$ dimensional simplex of edge $\beta$, defined as $\beta S_n=\{u| 0<u_1<\dots<u_n<\beta \}$, appear.
The extension of these integrals from $\beta S_n$, to $\beta B_n = (0,\beta)^n$ is done keeping the connected correlation function symmetric in the exchange of arguments similar to \cite{JoaoNicNic}. In particular, it holds that for every $F, F_i\in\mathcal{F}$
\[
\int_{\beta S_n} \di^n u\, \, \omega^\beta_C (F,\tau_{\mathrm{i} u_1}F_1(u_1),\dots ,\tau_{\mathrm{i} u_n}F_n(u_n) )
=
\frac{1}{n!}
\int_{\beta B_n} \di^n u\, \, \omega^\beta_C (F,\tau_{\mathrm{i} u_1}F_1(u_1),\dots ,\tau_{\mathrm{i} u_n}F_n(u_n) ).
\]

The state $\omega^\beta$ is quasi free, hence, the corresponding connected correlation functions admit an expansion in terms of a sum over  connected graphs whose edges correspond to the two-point function of the free theory (the operator $B_{\pm}$) and the vertices correspond to the various arguments in the connected correlation functions. This follows from the result of Proposition \ref{pr:connected-simple-graphs} considering the further expansion of $(e^{D_{ij}}-1)$ in powers of $D_{ij}$. With this in mind, each edge replaces two fields $\Psi$ and $\Psi^*$ present in different vertices and no fields $\Psi$ or $\Psi^*$ can remain in the final expressions of $\omega^{\beta A}(\Psi(f))$,  $\omega^{\beta A}(\Psi^*(h))$ and $\omega^{\beta A}(\Psi(f)\Psi^*(h))$.

The graphs in this expansion are connected, and thus there are simplifications in the expansion of 
$\omega^{\beta A}(\Psi(f))$ and 
$\omega^{\beta A}(\Psi^*(h))$
in terms of a sum over connected correlation functions evaluated at 
$(\Psi(f);\tau_{\mathrm{i}u_1}Q_{A}(u_1);\dots \tau_{\mathrm{i}u_n}Q_A(u_n))$ or 
$(\Psi^*(h);\tau_{\mathrm{i}u_1}Q_{A}(u_1);\dots \tau_{\mathrm{i}u_n}Q_A(u_n))$, and the subsequence decomposition of the entries  $Q_A=Q_{A,1}+Q_{A,2}$. In particular, the factor $Q_{A,1}$ appears one and only one time in each connected correlation function and the remaining factors need to be of the form $Q_{A,2}$.
With this considerations, we thus get the expressions $\omega^{\beta A}(\Psi(f))$ and $\omega^{\beta A}(\Psi^* (h))$ of the thesis. 

Similarly, in the expansion of 
$\omega^{\beta A}(\Psi(f)\Psi^*(h))$ as a sum over connected correlation functions and the subsequent decomposition of 
$Q_{A}$ in $Q_{A,1}+Q_{A,2}$, 
$Q_{A,1}$ appears either  two-times and the remaining factors are again all $Q_{A,2}$ 
or no elements corresponding to $Q_{A,1}$ are present. The elements where two factors $Q_{A,1}$ are considered are taken into account in $\omega^{\beta A}(\Psi(f))\omega^{\beta A}(\Psi^{*}(h))$. The remaining terms, involving $Q_{A,2}$ only, are then summed up in the truncated two-point function, hence to $\langle f, \Omega h\rangle$.




We finally analyze directly the expansion of the connected correlation functions in the expansion of $\omega^{\beta A}(\Psi(f)\Psi^*(h))$ which involves only $Q_{A,2}$. This correspond to the expansion of the operator $\Omega$ given in the thesis. To prove this we start again by Proposition \ref{pr:connected-simple-graphs} observing that in the expansion of
$
\omega_C^\beta( \Psi(f)\Psi^*(h); \tau_{\mathrm{i}u_1}Q_{A,2}(u_1);\dots; \tau_{\mathrm{i}u_n}Q_{A,2}(u_n))
$
in terms of connected graphs, the edges correspond to the thermal propagators given in Definition \ref{de:thermal-propagator} and the vertices to $A$. Furthermore, there are $n!$ terms (corresponding to the possible reordering of the $Q_{A,2}$) and each term is formed by $n+1$ vertices and $n+1$ edges. Once integrated in $(u_1,\dots u_n)$ over $(0,\beta)^n$ and exploiting the periodicity of the thermal propagators, all these terms are equal.
This furnishes the expression of $\Omega$ given in the thesis. 
\end{proof}
In the previous proposition, specifically in equation \eqref{eq:truncated-two-point}, we have seen that the series defining the operator $\Omega$ is actually a Dyson series written in terms of the thermal propagator $\Delta^{\beta}$, and $A$ acts as an external potential. We thus study the property of the operator $\Omega$ in the next proposition. Notice that similar results for the case of Fermi fields have been recently obtained in \cite{BruFrePin25} with the caveat that there $A$ is constant in $u$.

\begin{proposition}\label{pr:thermal-propagator}
Consider the thermal propagator $\Delta^\beta$ of the free theory, which is a linear operator on $\mathcal{H}\times L^2(0,\beta)$ that  is the fundamental solution of the equation
\[
\left(\frac{\partial}{\partial u} + K\right)\psi=0,
\]
periodic in $u$ and such that 
\[
\lim_{u\to0^+}\Delta^\beta(u) = \Omega_0.
\]

Consider $A(x,u)$ a function of $u$ and $x$ (smooth and compactly supported in $x$) such that $\beta \|A\|_\infty < 1$. The thermal propagator with the influence of the external potential ${V}_A$ is denoted by $\Delta^{\beta A}$ and it is the fundamental solution of
\[
\left(\frac{\partial}{\partial u} + K + A(u)\right)\psi=0,
\]
obtained as the Dyson series
\begin{equation}
\label{eq:thermal-interacting}
\Delta^{\beta A} = \sum_{n\geq 0} (-1)^n (\Delta^\beta A)^n \Delta^\beta.
\end{equation}
Its integral kernel can be used to reconstruct the truncated two-point function given in equation \eqref{eq:truncated-two-point} 
\[
\Omega = \lim_{u\to 0^{+}}\Delta^{\beta A} (0,u). 
\]
and it is equal to 
\begin{equation}\label{eq:DeltaO}
\Delta^{\beta A}(\overline{u},u)
=\mathcal{O}_1(\overline{u},u)
\end{equation}
where, for $u,\overline{u} \in [0,\beta)$ 
\begin{align*}
\mathcal{O}_{\lambda}(\overline{u},u) 
&=
\overline{\mathcal{T}}e^{-\int_{\overline{u}}^u (K+\lambda A(s))\di s}
\theta(u-\overline{u})
+
\overline{\mathcal{T}}e^{-\int_{\overline{u}}^\beta (K+\lambda A(s))\di s}
(1-\overline{\mathcal{T}}e^{-\int_0^\beta (K+\lambda A(s))\di s})^{-1 }
\overline{\mathcal{T}}e^{-\int_0^u (K+\lambda A(s))\di s}
\\
&=
\overline{\mathcal{T}}e^{-\int_{\overline{u}}^u (K+\lambda A(s))\di s}
\theta(u-\overline{u})
+
\sum_{n\geq 1}
\overline{\mathcal{T}}e^{-\int_{\overline{u}}^{n\beta+u} (K+\lambda A(s))\di s},
\end{align*}
in which $\overline{\mathcal{T}}$ denotes the anti time ordered product with respect to $u$. $\mathcal{O}_{\lambda}$ is thus anti time ordered exponential of $K+\lambda A$.
\end{proposition}

\begin{proof}

Consider the thermal propagator $\Delta^\beta(u)$
given in Definition \ref{de:thermal-propagator} and the corresponding covariance which promotes $\Delta^\beta(u)$ to an operator on $\mathcal{H}\times L^2(0,\beta)$ as in \eqref{eq:covariance} and it is such that
%
\[
\Delta^{\beta}\psi(x,u) = \int\di\overline{u} \Delta^{\beta}(\overline{u}-u)\psi(x,\overline{u}). 
\]
By direct inspection an inverse of this operator is $\partial_u + K$ and furthermore, the limit of its integral kernel $\Delta^\beta(u)$ for $u\to0^+$ is $\Omega_0$. 
Exploiting the periodicity of the thermal propagator, as in Proposition \ref{prop:one-point-two-point}, and considering equation \eqref{eq:truncated-two-point} we also have that 
\begin{align*}
\Omega 
&= \Delta^\beta(0) + 
\sum_{n\geq 1}(-1)^n
\int_{\beta {B}_n}\di^n u\,
\Delta^\beta(u_1) A(u_1)
\Delta^\beta(u_2-u_1)
A(u_2)
\dots 
\Delta^\beta(u_n-u_{n-1}) A(u_n) \Delta^\beta(-u_n).
\end{align*}
We may thus introduce the interacting thermal propagator as the result of the Dyson series
\begin{align*}
\Delta^{\beta A}(\overline{u},u)
=&
\Delta^\beta(u-\overline{u})  \\
+&\sum_{n\geq 1}(-1)^n
\int_{\beta {B}_n}\di^n u\,
\Delta^\beta(u_1-\overline{u}) A(u_1)
\Delta^\beta(u_2-u_1)
A(u_2)
\dots 
\Delta^\beta(u_n-u_{n-1}) A(u_n) \Delta^\beta(u-u_n).
\end{align*}
In the sense of  operators on functions in $\Sigma \times [0,\beta]$,  interpreting $A$ as a multiplicative operator, this is 
\begin{equation}
\label{eq:thermal-interacting}
\Delta^{\beta A} = \sum_{n\geq 0} (-1)^n (\Delta^\beta A)^n \Delta^\beta
\end{equation}
Comparing with the expansion of the truncated two-point function given in equation \eqref{eq:truncated-two-point} we have that 
\[
\Omega = \lim_{u\to 0^{+}}\Delta^{\beta A} (0,u). 
\]
We now discuss an expansion of this operator in terms of anti time ordered products.
Consider now $\mathcal{O}_\lambda(\overline{u},u)$.
we recall that
\begin{align*}
\mathcal{O}_{\lambda}(\overline{u},u) 
:=
\sum_{n\geq 0}
\overline{\mathcal{T}}
e^{-\int_{\overline{u}}^{u+n\beta}(K+\lambda A(s))\di s}.
\end{align*}
Exploiting the periodicity of $A(s)$, we have that 
\begin{align*}
\mathcal{O}_{\lambda}(0,u) 
=
(1-\overline{\mathcal{T}}e^{-\int_0^\beta (K+\lambda A(s))\di s})^{-1 }
\overline{\mathcal{T}}e^{-\int_0^u (K+\lambda A(s))\di s}.
\end{align*}
For non vanishing $\overline{u}$ with $u,\overline{u} \in [0,\beta]$, we have
\begin{align*}
\mathcal{O}_{\lambda}(\overline{u},u) 
=
\overline{\mathcal{T}}e^{-\int_{\overline{u}}^u (K+\lambda A(s))\di s}
\theta(u-\overline{u})
+
\overline{\mathcal{T}}e^{-\int_{\overline{u}}^\beta (K+\lambda A(s))\di s}
(1-\overline{\mathcal{T}}e^{-\int_0^\beta (K+\lambda A(s))\di s})^{-1 }
\overline{\mathcal{T}}e^{-\int_0^u (K+\lambda A(s))\di s}.
\end{align*}
Deriving with respect to $\lambda$ , for non vanishing $\overline{u}$,  
we obtain after a straightforward computation
\begin{align*}
\frac{d}{d\lambda} \mathcal{O}_\lambda(\overline{u},u)
&=
-\int_0^\beta \di s_1 
\mathcal{O}_\lambda(\overline{u},s_1)
A(s_1)
\mathcal{O}_\lambda(s_1,u)
\end{align*}
Computing the $n$-th order derivative we get ($n!$ times) the coefficients of the power expansion around $\lambda=0$ of $\mathcal{O}_\lambda(\overline{u},u)$. 
This coefficients coincides with  the contributions of the expansion of $\Delta^{\beta, \lambda A}(u,\overline{u})$.
We have proven that 
the  operator $\mathcal{O}_\lambda (\overline{u},u)$ has the same power expansion in terms of powers of $\lambda$ as $\Delta^{\beta \lambda A} (u,\overline{u})$. The power expansion is absolutely convergent and thus they coincide.
\end{proof}
We conclude this section, observing that in the limit of vanishing background $\phi_0\to0$, $\Phi=\Psi$, the truncated two-point function are equal to $\omega^{\beta A} (\Phi(f)\Phi^*(h)) = \langle f,\Omega h \rangle$ and $\omega^{\beta A}$ is quasi free.

    


\subsection{Relative partition function and relative entropy of the auxiliary theory}

We recall that, according to equation \eqref{eq:relative-entropy-basic}, to compute the relative entropy we have to evaluate
\begin{equation}\label{eq:rel-entropy}
\begin{aligned}
\mathcal{S}_{\phi_0}
&= 
\log{} 
(\omega^\beta
(U_A^{\phi_0}(i\beta))+\int_0^\beta \omega^\beta(\tau_{\mathrm{i}u }Q_A(u)) \di u
 \\
&= \sum_{n>1}(-1)^n
\int_{\beta S_n} \di^n u\, \omega^\beta_C(\tau_{\mathrm{i}u_1}Q_A(u_1),\dots, \tau_{\mathrm{i}u_n}Q_A(u_n))
\\
&= 
\sum_{n>1}(-1)^n\frac{1}{n!}
\int_{\beta {B}_n} \di^n u\, \omega^\beta_C(\tau_{\mathrm{i}u_1}Q_A(u_1),\dots, \tau_{\mathrm{i}u_n}Q_A(u_n))
\end{aligned}\end{equation}
the extension of the integrals from $\beta S_n$, the $n$ dimensional simplex of edge $\beta$,  to $\beta B_n = (0,\beta)^n$ is done by symmetry arguments \cite{JoaoNicNic}. 
Furthermore here $Q_A$ was given in \eqref{eq:QA}, we notice that this corresponds to compute a trace.
We have actually the following proposition

\begin{proposition}\label{prop:rel-entropy-0}
The relative entropy of $\omega^{\beta A}$ relative to $\omega^\beta$ as states on $\mathcal{F}$ and given in \eqref{eq:rel-entropy} for the case $\phi_0=0$ can be expanded as
\begin{equation}\label{eq:rel-entropy-S0}
\begin{aligned}
\mathcal{S}_0
&=
 \Tr_{u} (\sum_{n>1}\frac{(-1)^n}{n} (A \Delta^{\beta})^n \chi)
\\
&=
\Tr_{u} (A \int_0^1 \di \lambda  (\Delta^\beta-\Delta^{\beta, \lambda A})\chi).
\end{aligned}
\end{equation}
Where $A$ is a multiplicative operator of compact spatial support, $\chi$ is a cutoff function which is $1$ on the support of $A$ and $\Tr_{u}$ is for the  trace on $L^{2}(\Sigma \times [0,\beta])$.
\end{proposition}
\begin{proof}
This result follows directly from the expansion of the connected correlation functions in terms of connected graphs given in \eqref{eq:connnected-correlations-graph-expansion}, combining it  with the results of Proposition \ref{prop:one-point-two-point} valid for the case $\phi_0=0$ ($Q_{A,1}=0$)
and of Proposition \ref{pr:thermal-propagator} and the expansion of the interacting thermal propagator given in 
\eqref{eq:thermal-interacting}.
Notice that, the factor $1/n$ arises because we multiplied the factor $1/n!$ with the number of possible sequences of the first $n$ natural number which starts with $1$. This count to the possible ordering of the elements $\tau_{\mathrm{i}u_j}Q_A(u_j)$ present in the connected correlation functions keeping fixed the first one (at order $n$, this gives a factor $(n-1)!$). 
\end{proof}
Expanding the integrals in $u$, we have
\begin{align}\label{eq:rel-entropyThermalPropagators}
\mathcal{S}_0
&=
 {\Tr} (
 \int_0^{\beta} \di u A(u) \int_0^1 
\di \lambda
(
\Delta^{\beta}(0)
-\Delta^{\beta, \lambda A}(u,u) )\chi).
\end{align}
where now the trace is on $L^{2}(\Sigma)$ and $\Delta^{\beta, \lambda A}(u,v)$ is the integral kernel of the operator $\Delta^{\beta,\lambda A}$. 
We also observe that $\Delta^{\beta}(v)
-\Delta^{\beta, \lambda A}(u,u+v))$ is continuous in $v$ also near $0$ because $\Delta^{\beta}(0^+)-\Delta^{\beta}(0^-) =1$ and the same holds for $\Delta^{\beta A}(u,u+0^{\pm})$.
%
%
%
%
%
\begin{proposition}
Consider $A$ smooth and compactly supported, let $\chi\in C^\infty_0(\Sigma)$ be equal to $1$ on the support of $A$. The operator valued function  $\mathfrak{S}:u\mapsto A(\cdot,u) (\Delta^\beta(u)- \Delta^{\beta A}(u,u))\chi$ with domain $(0,\beta)$ and with range in the set of bounded linear operators on $L^2(\Sigma)$. 
It holds that for every $0<u<\beta$, if $A$ is sufficiently small $\mathfrak{S}$ is trace class.
\end{proposition}

\begin{proof}
To prove this claim we use a strategy closely related to the proof of Lemma 9 in \cite{BruFrePin25}. We observe in particular that 
%
%
 for every $u$, $\mathfrak{S}(u)$ is a bounded operator, because $A$ is a bounded multiplicative operator and both $\Delta^\beta(u)$ and $\Delta^{\beta A}(u)$ are bounded thanks to the fact that the spectrum of $K$ is strictly positive. 
%
%
%
Notice that for every $u\neq 0$, $u\in (0,\beta]$ we also have that for some constant $C$
\[
\| \chi \sqrt{\Delta^\beta(u)}\|_{HS}=\| \sqrt{\Delta^\beta(u)}\chi\|_{HS} \leq C\frac{1}{{u}^{\frac{3}{4}}} \|\chi\|_2 
\]
because $\chi$ is of compact support and the for ${u}\neq 0$ and thanks to the spectral properties of $K$.

Consider now
\[
\mathfrak{S}
=
\sum_{n\geq 1}(-1)^{n+1}
\int_{\beta {B}_n}\di^n u\,
A(u)\Delta^\beta(\beta+u_1-u) A(u_1)
\Delta^\beta(u_2-u_1)
A(u_2)
\dots 
\Delta^\beta(u_n-u_{n-1}) A(u_n) \Delta^\beta(u-u_n)\chi.
\] 
For every element of the domain $\beta {B}_n$ at least one element $y\in\{\beta+u_1-u, u_2-u_1, \dots , u-u_n \}$
is larger than $\beta/{(n+1)}$. Decomposing the corresponding $\chi\Delta^\beta(u)\chi = \chi\sqrt{\Delta^\beta(u)}\sqrt{\Delta^\beta(u)}\chi$, we get that  for every $n$ the integrand is trace class because it is a product of two Hilbert Schmidt operators and bounded operators. 
The estimates given above permit to take the integrals and the sums if $A$ is sufficiently small. 
\end{proof}

\begin{proposition}\label{pr:rel-entropy-condensate}
The contribution 
due to the condensate to the 
relative entropy of $\omega^{\beta A}$ relative to $\omega^\beta$ as states on $\mathcal{F}$  is 
\[
\mathcal{S}_{\phi_0}-\mathcal{S}_0
= \sum_{n > 0} (-1)^{n+1}    \langle\phi_0 A, ( \Delta^{\beta}A)^{n-1} \Delta^\beta A \phi_0 \rangle_u
\]
where the scalar product is over $L^2( \Sigma\times [0,\beta])$.
In terms of the thermal propagator of the interacting theory, it can be written as
\[
\mathcal{S}_{\phi_0}-\mathcal{S}_0
=    \langle \phi_0 A, \Delta^{\beta A}  A \phi_0\rangle_u
\]
Furthermore $\Delta^{\beta A} $ is invertible and $(\Delta^{\beta A})^{-1} = \frac{\partial}{\partial u} + K +A $ and thus
\begin{align} \label{eq:rel-entropy-condensate}
\mathcal{S}_{\phi_0}-\mathcal{S}_0
&=
\langle  {A}\phi_0, \phi_0\rangle_u
-
\langle  \phi_0, 
K\phi_0\rangle_u
+
\langle  K \phi_0, 
\Delta^{\beta A}
K \phi_0\rangle_u .
\end{align}
\end{proposition}

\begin{proof}
The first expression follows directly from the expansion of $\log\omega^\beta(U_A(i\beta))$ in terms of connected correlation functions given in \eqref{eq:logZA} after discarding the zeroth order contribution.
The second expression for $\mathcal{S}_{\phi_0}-\mathcal{S}_0$ can be obtained owning the form of $\omega^{\beta A}$ given in \eqref{eq:thermal-interacting}.
Finally, we observe by direct inspection that 
\[
{\Delta^{\beta A}}^{-1} = \frac{\partial}{\partial u} + K +A
\]
hence, 
\begin{align*}
\mathcal{S}_{\phi_0}-\mathcal{S}_0
&=
\langle  {A}\phi_0, \phi_0\rangle_u
- 
\langle  A \phi_0, 
\Delta^{\beta A} (\frac{\partial}{\partial u}+K)
\phi_0\rangle_u
\\
&= 
\langle  {A}\phi_0, \phi_0\rangle_u
-
\langle  \phi_0, 
(\frac{\partial}{\partial u}+K)\phi_0\rangle_u
+
\langle  (\frac{\partial}{\partial u}+K)\phi_0, 
\Delta^{\beta A}
(\frac{\partial}{\partial u}+K)\phi_0\rangle_u 
\end{align*}
where we recall that the derivative in $u$ of $\Delta^{\beta A}$ is taken in the distributional sense. 
Finally, observing that $\phi_0$ is constant in $u$ we get the thesis.
\end{proof}
If $\phi_0$ is a stationary point of the classical Hamiltonian of the problem and if $A$ is of compact spatial support,  the previous analysis can be further improved considering $\chi_1$ and $\chi_2$ which  are  two smooth spatial cutoff functions which are positive, of compact spatial support and equal to $1$ on the domain of $A$ (domain of g).
With this $\chi_1$ and $\chi_2$ at disposal we have that the result of the previous proposition can be written in the following way
\[
\mathcal{S}_{\phi_0}-\mathcal{S}_0
=    \langle A\chi_1 \phi_0  ,\Delta^{\beta A}  A \chi
_2\phi_0\rangle_u
=    \langle  {A}\chi_1 \phi_0, \chi_2\phi_0\rangle_u
-
\langle  \chi_1 \phi_0, 
K\chi_2 \phi_0\rangle_u
+
\langle  K \chi_1 \phi_0, 
\Delta^{\beta A}
K \chi_2\phi_0\rangle_u .
\]


\section{Path integral representation}\label{se:5}

Before discussing the evaluation of $A$ on the Gaussian state with covariance $v$, it is useful to discuss the path integral representation of $\Omega$, of $\mathcal{S}_0$ and of $\mathcal{S}_{\phi_0}- \mathcal{S}_0$.
We start observing that $\Omega_0$ can be given in terms of a suitable integral over the Wiener measure. 
We shall then use the Feynman-Kac formula to take into account the external $x,u$ dependent potential $A(x,u)$ in the various propagators.
The theory we are considering is non relativistic and the description of $\langle f,\Omega g\rangle$ by means of a path integral is not plagued by the inconsistencies present in the case of relativistic theories. More precisely the measure over paths in the path integral representation of the integral kernel of $\langle f,\Omega g\rangle$ is a variant of the well defined Wiener measure. 

Notice also that the non commutativity of the $K$ with $A$ can be ignored in the path integral essentially because of an application of the Trotter product formula used in the derivation of the measure in the path integral representation of the integral kernel of $\Omega$, this plays a crucial role in the evaluation of the Gaussian field $A$ with $e^{\Gamma}$.

For vanishing $A$, with $m=1$ and 
with chemical potential $\mu<0$ the operator $K$ is such that
\[
\hat{K} = \frac{p^{2}}{2} -\mu
\]
The integral kernel of the corresponding operator $\Omega_0$ is
\[
\Omega_0(x,y) = \int_0^{\infty} \di t \; W(t,\beta,\mu)\int   \di \mu^{y,t}_{x,0}(\omega)
\]
where $\mu^{y,t}_{x,0}(\omega)$ is the Wiener measure of the path $\omega$ joining $x$ and $y$ in a time $t$. 
Furthermore, $W(t,\beta,\mu)$ is a suitable positive weight.
To obtain this result and the form the weight, we use the fact that the fundamental solution of the heat equation in Fourier domain is proportional to $e^{\frac{-t p^2}{2}}$ and hence 
\[
\hat{\Omega}_0 = \frac{1}{1-e^{-\beta \hat{K}}} = \sum_{n\geq 0} e^{- \beta n \hat{K}} = \sum_{n\geq 0} e^{-\beta n \frac{p^{2}}{2} + \beta n \mu}
\]
We get that the integral kernel of 
$\frac{e^{-u K}}{1-e^{-\beta K}}$ 
is
\begin{align*}
\Omega_0(x,y;u) 
&= 
\sum_{n\geq 0}
\int_0^{\infty} \di t \; e^{(\beta n+u) \mu } \delta{(t-\beta n-u)} \int   \di\mu^{y,t}_{x,0}(\omega)
\\
&= 
\sum_{n\geq 0} e^{(\beta n+u) \mu } \int   \di\mu^{y,\beta n+u}_{x,0}(\omega)
\end{align*}
We can now take into account the effect of $A$ by means of the {\bf Feynman-Kac} formula. 
For $u=0$ it reads
\begin{align}\label{eq:omega}
\Omega(x,y) 
&= 
\sum_{n\geq 0} e^{\beta n \mu } \int   \di\mu^{y,\beta n}_{x,0}(\omega) e^{- \sum_{0\leq j\leq n-1}\int_0^{\beta} A(s,\omega(s+j\beta)) \di s }.
\end{align}

\subsection{Path integral representation of the interacting propagator}

Recalling \eqref{eq:DeltaO} and Proposition \ref{pr:thermal-propagator} 
we have that for $u,\overline{u} \in [0,\beta)$ and for $x,y \in \Sigma$ the Feynman-Kac like formula for the thermal propagator perturbed with the external potential $A$ can be given as
\begin{equation}\label{eq:path-integral-thermal}
\begin{aligned}
\Delta^{\beta ,  A}(x,u;  y, \overline{u}) 
=&  
e^{(u-\overline{u}) \tilde{\tilde{\mu}}}  \int \di\mu^{x, u}_{y,\overline{u}}(\omega) e^{\int_{\overline{u}}^{u} A(\omega(u),s) \di s } 
\theta(u-\overline{u})
\\
&+\sum_{n\geq 1} e^{(\beta n+(u-\overline{u})) \tilde{\tilde{\mu}}}  \int \di\mu^{x, n\beta + u}_{y,\overline{u}}(\omega) e^{\int_{\overline{u}}^{n\beta +u} A(\omega(u),s) \di s } 
\end{aligned}
\end{equation}
where $A$ is extended to $\Sigma \times \mathbb{R}$ by periodicity. 
With this expression at disposal, we can now represent as suitable path integrals, the elements which appears in $T_1,T_2$ namely the 
relative entropy $\mathcal{S}_0$ present in equation \eqref{eq:rel-entropyThermalPropagators} of in Proposition \ref{prop:rel-entropy-0} 
and $\mathcal{S}_{\phi_0} - \mathcal{S}_0$ in equation \eqref{eq:rel-entropy-condensate} of Proposition \ref{pr:rel-entropy-condensate}.
More precisely, recalling that $W_A =  -\omega^\beta({V}_A) + \mathcal{S}_{\phi_0} = T_0 +T_1+T_2$, 
where $T_0= -\omega^\beta({V}_A)$,
$T_1=\mathcal{S}_{\phi_0}-\mathcal{S}_0$,
and that $T_2=\mathcal{S}_0$ we have
\begin{align}\label{eq:Sphi-propagator}
\mathcal{S}_{\phi_0}  &=T_{1} + T_2
\qquad T_1=
\langle \phi_0 A \Delta^{\beta A}  A \phi_0\rangle_u
\qquad 
T_2= 
\Tr_{u} (A \int_0^1 \di \lambda  (\Delta^\beta-\Delta^{\beta, \lambda A})\chi).
\end{align}
also these objects admits a representation in terms of path integrals obtained from the path integral representation of the propagators and the Feynman-Kac given in equation \eqref{eq:path-integral-thermal}.
%
%


%
%

\section{Convergence of the interacting equilibrium state constructed above}\label{se:6}

Up to the evaluation of $A$ in the Gaussian state with covariance $ v(x-x')\delta(u-u')$ we have finiteness of $\omega^{\beta A}$. 
We shall now discuss how the evaluation of the auxiliary Gaussian field  leads to a finite theory.
We start observing that the direct application of $\Gamma$ to polynomials of $A$ has a complex combinatoric and in particular the number of terms we obtain at order $2N$ in $A$ is $(2N-1)!!$, see the appendix \ref{ap:counting}. 
This large number does not permit a direct analysis of the convergence of the power series in $A$ present in $\omega^{\beta A}$.
The best we can hope is in the proof of Borel summability of the obtained expansion as e.g. in \cite{FroehlichKnowlesSchleinSohinger2} for the case of vanishing condensate. 
We observe at the same time that the following relations hold
\[
\left.e^{\Gamma} \prod_i e^{\langle A \rangle_{b_i}} \right|_{A=0}  = e^{-\sum_{i,j} \frac{1}{2} \tilde{v}(b_i,b_j)}
\]
where 
\begin{equation}\label{eq:<A>b}
\langle A\rangle_b = \int_{\Sigma \times [0,\beta]} A(x,u) b(x,u)\di^3 x\di u
\end{equation}
with $b$ a function on $\Sigma \times [0,\beta]$, we shall use below the same notation for $b$ being generalized functions. 
Furthermore, 
\[
\tilde{v}(b_i,b_j) = \int \di^3 x \di^3 x' du du' \int v(x-x')\delta(u-u') b_i(x,u) b_j(x',u').
\]
These relations reminds the commutation relations of an (exponentiated) Weyl algebra and permits to simplify the combinatoric present in the power expansions in $A$, actually the contribution of an infinite number of monomials in $A$ are obtained in a single finite object. 

Before applying $\Gamma$ it is thus advisable to expand the various elements in powers of exponentials of $A$. 
Notice that the Feynman-Kac formula used to expand the thermal propagator $\Delta^{\beta A}$ e.g. in equation \eqref{eq:path-integral-thermal} gives by construction the sought expansion in exponentials of $A$.
We furthermore have also the following

\begin{proposition}
The action of $e^{\Gamma}$ can be switched with the path integrals  which shows up in the expansion of the thermal propagator and the various expression involving the relative entropy. Namely
\[
C:=e^{\Gamma} \int \di \mu_{{x},0}^{{y},T}(\omega) e^{\int_0^{T} g(\omega(s))A(\omega(s),s)\di s} = \int \di \mu_{{x},0}^{y,T}(\omega) e^{\Gamma} e^{\int_0^{T} g(\omega(s))A(\omega(s),s)\di s}
\]
\end{proposition}
\begin{proof}
Consider $\mu_{{x},0}^{{y},T}$ the measure over paths present in $C$.
Then, we split the parametrizing interval $[0,T]$ into subintervals:
\begin{equation}
\bigcup_{i = 1}^n [t_{i-1},t_i] = [0,T]
\end{equation}
where $t_0 = 0$ and $t_n = T$. Then, we define the following finite measures on $(\mathbb{R}^3)^{n-1}$:
\begin{equation}
	\di \mu_n({x}_1, \ldots, {x}_{n-1}) = \frac{1}{Z_n}\prod_{i=1}^n p(t_{i-1} - t_i,{x}_{i-1},{x}_i)\di^3 {x}_1 \cdots \di^3 {x}_{n-1}
\end{equation}
where $Z_n$ is some normalisation constant, ${x}_0 = {x}$ and ${x}_{n} = {y}$ and where 
\begin{equation}
	p(t,{x},{y}) := \frac{1}{(2 \pi)^{3/2}} \exp\left( -\frac{|{x} - {y}|^2}{2t} \right), \qquad t\in \mathbb{R}, x,y\in\mathbb{R}^3
\end{equation}
denotes the heat kernel.
Then, by Kolmogorov theorem it holds that:
\begin{equation}
	\mu_{{x},0}^{{y},T} = \lim_{n \to \infty} \mu_n.
\end{equation}
Therefore, we can rewrite:
\begin{equation}
	C=e^{\Gamma} \lim_{n \to \infty} \int \di \mu_n({x}_1,\ldots, {x}_{n-1}) e^{\sum_{i=1}^n g({x}_i)A({x}_i,t_i)\Delta t}.
\end{equation}
where $\Delta t = T/n$ and $t_i = i \Delta t$. Therefore, the question becomes whether we can switch $e^{\Gamma}$ with the $n \to \infty$ limit first and with the integration on $(\mathbb{R}^3)^{n-1}$ after. First of all we notice that since $A$ is bounded (and $\sup_{{x} \in \mathbb{R}^3} g({x}) =1$), we can switch $e^{\Gamma}$ with the integration over $(\mathbb{R}^3)^{n-1}$ using dominated convergence. So we only have to check the interchange with the limit $n \to \infty$. However, this also holds since the limit exists pointwise (converges to the Wiener measure) and the functional derivative of the elements in the sequence is uniformly bounded again because of the boundedness of $A$. Therefore getting the thesis.
\end{proof}


If we now apply $e^{\Gamma}$ before taking the path integral we have 
\begin{align*}
e^{\Gamma}\Omega(x,y) 
&= 
\sum_{n\geq 0} e^{\beta n \mu } \int   \di\mu^{y,\beta n}_{x,0}(\omega) e^{-  \frac{1}{2}
\sum_{0 \leq i,j \leq {n-1}}
\int_0^{\beta} \int_0^{\beta} 
v(\omega(s+j\beta)-\omega(s'+i\beta))\delta{(s-s')}
 \di s \di s' } 
 \\
&= 
\sum_{n\geq 0} e^{\beta n \mu } \int   \di\mu^{y,\beta n}_{x,0}(\omega) e^{-  \frac{1}{2}
\sum_{0 \leq i,j \leq {n-1}}
\int_0^{\beta}  
v(\omega(s+j\beta)-\omega(s+i\beta))
 \di s } 
 \end{align*}

Notice that for every path $\omega$ and for every $s$
\[
E=\sum_{0\leq i,j\leq n-1} v(\omega(s+j\beta) - \omega(s+i \beta))  
\]
is positive because the linear operator obtained by means of a convolution from $v$ defines a positive product
\[
\langle \overline\Psi,v * \Psi\rangle
\]
for $\Psi \in L^2(\mathbb{R})$.
For 
$\Psi_\epsilon = \sum_{0\leq j\leq n-1}\delta_{\epsilon}(x-\omega(s+j\beta))$
where $\delta_\epsilon$ are suitable mollification of the delta function, hence
for ever $\epsilon $, $\langle \overline \Psi_\epsilon,v*\Psi_\epsilon\rangle \geq 0$ and the same holds in the limit of vanishing mollification, namely
\[
E = \lim_{\epsilon \to 0}\langle \overline \Psi_\epsilon,v*\Psi_\epsilon \rangle \geq 0.
\]
The Wiener measure is also positive, hence, if $\mu\leq 0$ we have that 
\[
e^{\beta \mu n} 
e^{-  \frac{1}{2}
\sum_{0 \leq i,j \leq {n-1}}
\int_0^{\beta}  
v(\omega(s+j\beta)-\omega(s+i\beta))
 \di s } \leq 1
\]
hence we get the following bound for the two-point function
\begin{align*}
e^{\Gamma} \Omega(x,y) 
& \leq \Omega_0 (x,y)
\end{align*}

% ...

\subsection{Bounds satisfies by the partition function of the interacting equilibrium state}

In this subsection we discuss how to apply $e^{\Gamma}$ to $Z_A = e^{W_A}$ and to $S Z_A$ used to obtain the two-point function of the interacting theory in the equilibrium state we are considering as discussed in Proposition \ref{pr:what-we-compute}. 

We shall in particular make use of the Feynman-Kac formula to get an expansion of the thermal propagator in terms of  exponentials of suitably smeared $A$. 
Hence, we aim to expand the various elements of $Z$ and of $\left. e^{\Gamma}S Z_A\right|_{A=0}$ as a power series in the interacting thermal propagators $\Delta^{\beta \lambda A}$ and $\Delta^{\beta A}$ which appears in $\mathcal{S}_{\phi_0}$ and we discuss the absolute convergence of the obtained series. 
Up to a redefinition of the zeroth-order term, this is realized replacing $\Delta^{\beta \lambda A}$ by
$\xi \Delta^{\beta \lambda A}$
for various $\lambda$ in $\mathcal{S}_{\phi_0}$, 
expanding in powers of the {\bf auxiliary parameter} $\xi$ and estimating the truncated series for $\xi=1$.
To be more precise, the parameter $\xi$ appears in $W_A$ in the following way.
\begin{equation}\label{eq:WAxi}
W_A^\xi = -\omega^\beta{({V}_A)} + \Tr_{u} (g A \int_0^1 \di \lambda  (\Delta^\beta-\xi \Delta^{\beta, \lambda g A}))   +
\langle  g{A}\phi_0, \phi_0\rangle_u
+
\xi 
\langle  K \phi_0, 
(
\Delta^{\beta gA}
-\Delta^\beta)
K \phi_0\rangle_u 
\end{equation}
where we used $\mathcal{S}_0$ given in \eqref{eq:rel-entropy-S0} and $\mathcal{S}_{\phi_0}-\mathcal{S}_0$ given in \eqref{eq:rel-entropy-condensate} and where $\xi$ appears also in the last term proportional to $\Delta^\beta$ to simplify the estimate we are going to discuss.
The parameter $\xi$ will appear in the corresponding explicit expansion given in terms of the path integrals respectively as in \eqref{eq:Sphi-propagator}.
%
%

In view of the path integral expression for $\Omega$ in the following subsection, we analyze
the convergence of the power series in $\xi$ of the following expression.
\[
Z^\xi  = 
\left. e^{\Gamma } e^{W_A^\xi} \right|_{A=0}
\]
and of
\[
\left. e^{\Gamma }e^{\langle A \rangle_b} e^{W_A^\xi} \right|_{A=0} 
\]
for some function $b$ on $\Sigma \times [0,\beta]$. 

Recalling $W_A^\xi = T_0^\xi + T_{1}^\xi+T^\xi_2$ we start considering separately the contributions due to $T_2=\mathcal{S}_0$ and to $T_{1} = \mathcal{S}_{\phi_0}-\mathcal{S}_0$ in two separate subsection. We eventually combine the results.



\subsubsection{Bounds of the relative partition function valid for vanishing condensate $\phi_0=0$}

In this subsection we evaluate the contribution to $Z^\xi=e^{\Gamma}e^{W_A^\xi}|_{A=0}$ due to $W_{A,0}^\xi = T+T_0=\mathcal{S}_0-\omega^{\beta}(V_A)$ in $W_A$.
This is the contributions to the relative partition function which survives when the background value of the field $\Phi$ vanishes, namely when condensate is not considered ($W_{A,0}^{\xi}=\left. W_A^\xi\right|_{\phi_0=0}$).

We thus start considering 
\[
e^{\Gamma} e^{\langle A \rangle_b} e^{W_{A,0}^\xi}.
\]

We actually have the following proposition.

\begin{proposition}\label{prop:bounds-no-condenstate}
Consider 
\[
e^{\Gamma} e^{\langle A \rangle_b} e^{W_{A,0}^\xi}
\]
and
\[
\frac{e^{\Gamma} e^{\langle A \rangle_b} e^{W_{A,0}^\xi}}{e^{\Gamma} e^{W_{A,0}^\xi}}
\]
as power series in $\xi$ (powers of $\Delta^{\beta, \lambda A}$). 
The following absolute bounds hold uniformly in the order of truncation of the series for $\xi \leq 1$
\[
|e^{\Gamma} e^{\langle A\rangle_b}e^{W_{A,0}^\xi}|
\leq 
\exp E
,
\qquad
|e^{\Gamma} e^{W_{A,0}^\xi}|
\geq 
2-\exp E
,
\qquad
\frac{|e^{\Gamma} e^{\langle A\rangle_b}e^{W_{A,0}^\xi}|}{|e^{\Gamma}e^{W_{A,0}^\xi}|} \leq 
\frac{\exp E}{2-\exp E}
\]
where
\[
E=\frac{\tilde{{C}} {v(0)}^{\frac{1}{2}} \|g\|_1}{\beta}\xi
\]
hence, if ${E} < \log2$ the two series are absolutely convergent.
For $\xi=1$ this corresponds to have
\[
\frac{\tilde{{C}} {v(0)}^{\frac{1}{2}} \|g\|_1}{\beta} < \log 2
\]
\end{proposition}

\begin{proof}
    The proof is obtained combining the  steps obtained below through various lemmata.

In the evaluation of the expectation values in the thesis of the proposition we use  the following lemma. Which permits to discard the fact that $A$ depends on the imaginary time $u$ in $W_{A,0}^\xi$.   
\begin{lemma}\label{lm:lemma-analytic-s}
Consider
\begin{align}\label{eq:stepT}
T_2^\xi &= \int \di^3 x  \int_0^\beta \di s   g{A}(x,s) \left( 
 \sum_{n\geq 1} e^{\beta n \tilde{\tilde{\mu}} } \int   \di\mu^{x,\beta n}_{x,0}(\omega) 
 \int_0^1 \di  \lambda
 \left(1-\xi e^{-  
\lambda  \sum\limits_{j=0}
^{n-1}
\int_0^{\beta} gA(\omega(s_j+j\beta,s_j+s))
 \di s_j }
\right) \right)
\end{align}
and
\begin{align}\label{eq:stepTs}
T_{2,s}^\xi &= \int \di^3 x  \int_0^\beta \di s   g{A}(x,s) \left( 
 \sum_{n\geq 1} e^{\beta n \tilde{\tilde{\mu}} } \int   \di\mu^{x,\beta n}_{x,0}(\omega) 
 \int_0^1 \di  \lambda
 \left(1-\xi e^{-  
\lambda  
\sum\limits_{j=0}^{n-1}
\int_0^{\beta} gA(\omega(s_j+j\beta),s_j+\beta)
 \di s_j }
\right) \right)
\end{align}
For every sufficiently regular functional $F$ of $A$, it holds that 
\[
\left. e^{\Gamma} F e^{T_2^\xi}\right|_{A=0} = \left. e^{\Gamma} F e^{T_{2,s}^\xi}\right|_{A=0}.
\]
\end{lemma}
\begin{proof}
Consider $T_2^\xi$ in \eqref{eq:stepT}
In order to remove $s$ in the exponential in $T_2^\xi$ , we proceed integrating by parts in $s$ and we get
\[
T_2^\xi = T_{2,s}^\xi+R
\]
in the integration by parts, only the boundary term (for $s=\beta)$ matters and gives $T_{2,s}$, the reminder $R$ vanishes when $e^\Gamma$ is applied. 
This happens because  there 
\[
\frac{d}{ds} e^{-  
\lambda 
\int_{0}^{n \beta} gA(\omega(s_1),s_1+s)
 \di s_1 } = 
-\lambda 
e^{-  
\lambda 
\int_{0}^{n \beta} gA(\omega(s_1),s_1+s)
 \di s_1 }
\int_{0}^{n \beta} \partial_{2}gA(\omega(s_2),s_2+s)
 \di s_2 
\]
and this kind of contributions vanishes when $A(s_2)A(u)$ is replaced by $v$ be means of $\Gamma$. Actually
there the derivatives of the delta functions popups. The latter vanishes because that derivative of the delta function is integrated against constant functions. 
\end{proof}
Lemma \ref{lm:lemma-analytic-s} permits to use $T_{2,s}^\xi$ at the place of $T_2^\xi$ and to discard the reminder $R$ in the evaluation of $\left. e^{\Gamma} e^{\langle A \rangle_b} e^{W_{A,0}^\xi}\right|_{A=0}$. 
Discarding the reminder $T_2^\xi=T_{2,s}^\xi+R$, we rewrite $T_{2,s}^\xi$ in the following way
\begin{align*}
T_{2,s}^\xi &= \int \di^3 x\int_0^\beta ds  g(x)\langle A \rangle_{\delta_{x,s}} 
\sum_{n\geq 1} \int \di {\tilde{W}}_{n,x}(w)
e^{\beta \tilde{\tilde{\mu}} n} (1-\xi e^{-\langle A\rangle_{h_w}})
\\
&= \langle \int_0^\beta \di s {A}(x,s), g(x)  
\sum_{n\geq 1} \int \di {\tilde{W}}_{n,x}(w)
e^{\beta \tilde{\tilde{\mu}} n} (1-\xi e^{-\langle A\rangle_{h_w}})
\rangle
\end{align*}
where, here recalling equation \eqref{eq:<A>b}, $\delta_{x,s}$ is the Dirac delta centered in $(x,s)$ 
\[
\di {\tilde{W}}_{n,x}(w) =  \di\mu^{x,\beta n}_{x,0}(\omega)  \di\lambda 
\]
$\lambda \in [0,1]$
and $w=(\omega, \lambda)$.
Furthermore, this holds 
for $h_{w}$ 
which is a suitable linear combination of $A$ smeared along the paths and multiplied by $\lambda$. 
Hence combining these observations, 
\[
W_{A,0}^\xi = \mathcal{S}_0-\omega^\beta(V_A) =  T_{2,s}^\xi-\omega^{\beta}(V_A) +R
\]
and, discarding the reminder $R$, we have that 
\begin{align*}
W_{A,0}^\xi &= -\sum_{n\geq 1}
\xi \int \di s \int \di^3 x g(x)\langle A \rangle_{\delta_{x,s}} 
\int \di {\tilde{W}}_{n,x}(w)
e^{\beta \tilde{\tilde{\mu}} n} e^{-\langle A\rangle_{h_w}}
\end{align*}
where we used the fact that
\[
\omega^\beta(V_A)=
\sum_{n\geq 1}
\int \di s \int \di^3 x g(x)\langle A \rangle_{\delta_{x,s}}
\int \di {\tilde{W}}_{n,x}(w)
e^{\beta \tilde{\tilde{\mu}} n} 
=
\sum_{n\geq 1}
C\frac{ e^{\beta \tilde{\tilde{\mu}} n} }{(\beta n)^{\frac{3}{2}}}
\langle \tilde{A} \rangle_g = 
\frac{\tilde{C}(\beta \tilde{\tilde{\mu}})}{\beta^{\frac{3}{2}}}
\langle \tilde{A} \rangle_g
\]
In this expressions we used the notation introduced in \eqref{eq:<A>b}, furthermore we
 denoted by $\tilde{A}(x) =  \int_0^\beta \di s A(x,s)$
 and  we used 
the following Lemma. 
\begin{lemma}\label{lm:norm-state}
It exists a suitable positive constant $C$ such that  
\begin{equation}\label{eq:norm-measure}
\sum_{n\geq 1} e^{\beta \tilde{\tilde{\mu}} n } \int \di {\tilde{W}}_{n,x}
=
\sum_{n\geq 1} e^{\beta \tilde{\tilde{\mu}} n }
 \int \di\mu^{x,\beta n}_{x,0}(\omega) \int_0^1 \di \lambda 
=
C \sum_{n\geq 1} \frac{e^{\beta \tilde{\tilde{\mu}} n}}{(\beta n)^{\frac{3}{2}}} 
= \frac{C}{\beta^{\frac{3}{2}}} \mathrm{Li}_{\frac{3}{2}}(e^{\beta \tilde{\tilde{\mu}}}) 
 = \frac{\tilde{C}(\beta \tilde{\tilde{\mu}})}{\beta^{\frac{3}{2}}}
\end{equation}
where $\mathrm{Li}_s(y) = \sum_n \frac{y^n}{n^s}$ is the polylogarithm function and it is finite in $y=0$ for $s>1$, it is monotonically decreasing in $y$.
The result of this integral does not depend on $x$.
\end{lemma}
\begin{proof}
To prove this claim we observe that 
\[
\int \di {\tilde{W}}_{n,x} = \frac{1}{(2\pi)^3} \int e^{- \beta n \frac{p^2}{2}} \di ^3p = \frac{4\pi}{(2\pi)^3}  \sqrt{\frac{\pi}{2}} \frac{1}{(\beta n)^{\frac{3}{2}}}
\]
and in particular it does not depend on $x$.
\end{proof}

Now we evaluate $e^{W_{A,0}^\xi}$
where as before $\tilde{A}(x) = \int_0^\beta A(x,s)\di s$ and we discard $R$ because it will not contribute to the estimates thanks to Lemma \ref{lm:lemma-analytic-s}
\[
e^{W_{A,0}^\xi}=
\sum_{k\geq 0} \frac{(-1)^k \xi^k}{k!}\left(\langle \tilde{A}(x), g(x)  
\sum_{n\geq 1} \int \di {\tilde{W}}_{n,x}(w)
e^{\beta \tilde{\tilde{\mu}} n} e^{-\langle A\rangle_{h_w}}
\rangle
\right)^k
\]
\[
e^{W_{A,0}^\xi}=
\sum_{k\geq 0} \frac{(-1)^k \xi^k}{k!}\left(\int \di^3 x \int \di t e^{-t \tilde{A}(x)} g(x)  f(t)
\sum_{n\geq 1} \int \di {\tilde{W}}_{n,x}(w)
e^{\beta \tilde{\tilde{\mu}} n} e^{-\langle A\rangle_{h_w}}
\right)^k
\]
and
\[
e^{\Gamma} e^{-\langle A \rangle_b}e^{W_A^{0}}=
e^{\Gamma} e^{-\langle A \rangle_b}
\sum_{k\geq 0} \frac{(-1)^k \xi^k}{k!}\left(\int \di^3 x \int \di t e^{-t \tilde{A}(x)} g(x)  f(t)
\sum_{n\geq 1} \int \di {\tilde{W}}_{n,x}(w)
e^{\beta \tilde{\tilde{\mu}} n} e^{-\langle A\rangle_{h_w}}
\right)^k
\]
where
\[
f(t) =\frac{\partial}{\partial t} \delta(t)
\]
is the Fourier transform of
\[
\hat{f}(p) =  (-i p) .
\]
In order to apply $e^{\Gamma}$ 
we rearrange it in the following way
\begin{align*}
e^{\Gamma} e^{-\langle A \rangle_b}e^{W_{A,0}^\xi}=
\sum_{k\geq 0} \frac{(-1)^k \xi^k }{k!}
\prod_{j=1}^k
&\left(
\int \di^3 x_j  g(x_j)  
\int \di t_j
f(t_j)
\sum_{{n}_j\geq 1}
 \int dW_{n_j,x_j}(w_j)
e^{\beta \tilde{\tilde{\mu}} n_j}
\right) \cdot I
\end{align*}
where the integrand is
\begin{align*}
I=e^{\Gamma}
\left( e^{-\langle A \rangle_b}
\prod_{j=1}^k
 e^{-t_j \tilde{A}(x_j)}  
e^{-\langle A\rangle_{h_{w_j}}}
\right)
\end{align*}
we study the integrand (second line in the previous formula) and 
denoting by 
\[
\tilde{A}(x) = \langle A \rangle_x = \int \di u \int \di^3 y A(x,u )\delta(y-x)
\]
we have for $k>0$
\begin{align*}
I &= e^{\Gamma}
\left( e^{-\langle A \rangle_b}
\prod_{j=1}^k
 e^{-t_j \tilde{A}(x_j)}  
e^{-\langle A\rangle_{h_{w_j}}}
\right)
=
e^{\Gamma} e^{-\langle A\rangle_{b}
-\sum_{j}  t_j \langle A \rangle_{x_j}
-\sum_{j} \langle A \rangle_{h_{w_j}}}
\\
&=e^{-\frac{1}{2}\tilde{v}( -b
-\sum_{j}  t_j \delta_{x_j}
-\sum_{j} h_{w_j},
-b
-\sum_{j}  t_j \delta_{x_j}
-\sum_{j} h_{w_j})}
\end{align*}
to compute the following integral in $e^{\Gamma}e^{\langle A \rangle_b} e^{W_{A}^0}$ we start considering
\begin{align*}
B=\int \left(\prod_{j=1}^k \di t_j f(t_j)\right)I(t)   
\end{align*}
we assume that the various $x_j$
are chosen in such a way that
the matrix with entries $\eta_{ij}=\tilde{v}(\delta_{x_i},\delta_{x_j})$ is strictly positive. By Hadamard inequality it holds that $\det \eta \leq (\beta v(0))^n$ where $\beta v(0)= \eta_{jj}$.
This estimate ceases to be valid on the subset $D_n$ of $\Sigma^{n}$ corresponding to its diagonals (where it exist at least two $i,j$ such that $x_i\neq x_j$). This set of points is however a zero measure set on $\Sigma^n$. In the next two Lemmata we find a bound uniform on $\Sigma^n \setminus D_n$ hence after integrating against $g^n$ we get a valid bound.

\begin{lemma}\label{lm:step1}
Consider 
\begin{align*}
B=\int \left(\prod_{j=1}^k \di t_j f(t_j)\right)I(t).   
\end{align*}
As a function on $(x_1,\dots , x_k)\in \Sigma^k$.
Almost everywhere on $\Sigma^k$
it holds that 
\[
| B | \leq  \frac{1}{(2\pi)^{\frac{k}{2}}}
 \int \di^{k} p  \sqrt{\det\eta^{-1}} \left(\prod_l | p_l|\right)
e^{-\frac{p_i \eta^{ij} p_j}{2} }
\]
where $\eta$ is a positive definite square $k\times k$ matrix whose entries are $\eta_{ij} = \tilde{v}(\delta_{x_1},\delta_{x_{j}})$ and $\eta^{ij}$ denotes the components of its inverse $\eta^{-1}$.
\end{lemma}

\begin{proof}

We compute the Fourier transform, and after some translation we have that 
\begin{align*}
B&=\int \left(\prod_{j=1}^k \di t_j f(t_j)\right)I
\\
&=\frac{1}{(2\pi)^{\frac{k}{2}}}
 \int \di^{k} p  \sqrt{\det \eta^{-1}} \left(\prod_l \mathrm{i} p_l\right)
e^{-\frac{p_i \eta^{ij} p_j}{2} }
e^{-\mathrm{i} p_i \eta^{ij} \tilde{v}(\delta_{x_j},h)} 
e^{\frac{1}{2} \tilde{v}(\delta_{x_i},h) \eta^{ij} \tilde{v}(\delta_{x_j},h)} 
e^{- \frac{1}{2}\tilde{v}(h,h)}
\end{align*}
where $h= -b+\tilde{g} 
-\sum_{j} h_{w_j}$.


Notice that, 
since $\tilde{v}$ is positive, by Cauchy Schwarz inequality
\[
e^{\frac{1}{2} \tilde{v}(\delta_{x_i},h) \eta^{ij} \tilde{v}(\delta_{x_j},h)} 
e^{- \frac{1}{2}\tilde{v}(h,h)}\leq 1,
\]
and, 
\[
|e^{-\mathrm{i} p_i \eta^{ij} \tilde{v}(\delta_{x_j},h)} |\leq 1
\]
thus concluding the proof.
\end{proof}
We observe that the estimate obtained in Lemma \ref{lm:step1} is uniform in $x_j$ for every $j$, furthermore it is also uniform on the various $b$s and $h$s. To simplify the obtained bound, we introduce the following

\begin{lemma}\label{lm:integral}
It holds that 
\[
\int \di^{k} p \sqrt{\det\eta^{-1}} \left(\prod_{j=1}^k |p_j|\right) e^{-\frac{1}{2}p_a\eta^{ab}p_b}
\leq \pi^{\frac{k}{2}}(\frac{1}{k^{k/2}}) \frac{\Gamma(k)}{\Gamma(\frac{k}{2})}(\beta v(0))^{\frac{k}{2}}.
\]
\end{lemma}
\begin{proof}
To get this bound, proceed as follows, consider $P$ the unitary matrix which diagonalizes $\eta^{-1}$ and denote by $\lambda_i^{-1}$ the elements on the diagonal obtained after diagonalization, (these are the eigenvalues of $\eta^{-1}$ and they are positive). 
We observe that bounding the determinant of $\eta$ by 
\[
\det \eta \leq \prod_j \eta_{jj} \leq \prod_j\tilde{v}(\delta_{x_j},\delta_{x_j}) \leq (\beta v(0))^k
\]
and  using the arithmetic mean  geometric mean inequality we have 
\[
\prod_{j=1}^k|p_{j}|
=\frac{1}{\sqrt{\det \eta^{-1}}}\prod_{j=1}^k|\frac{p_{j}}{\sqrt{\lambda_j}}|
\leq \frac{1}{\sqrt{\det \eta^{-1}}}\frac{1}{k^{\frac{k}{2}}} \left(\sum_j \frac{|p_j|^2}{\lambda_j}\right)^\frac{k}{2}
\leq\frac{(\beta v(0))^{\frac{k}{2}}}{k^{\frac{k}{2}}} \|p\|_2^k
\]
where the norm $\| p \|_2^2 = p_i\eta^{ij}p_j$.
Using coordinates which diagonalizes $\eta^{-1}$ in the Fourier domain, we furthermore obtain that 
\[
\int \di^k p \sqrt{\det\eta^{-1}} \|p\|_2^k  e^{-\frac{1}{2}\|p\|_2^2}
= \pi^{\frac{k}{2}}\frac{\Gamma(k)}{\Gamma(\frac{k}{2})}. 
\] 
Combining the two estimates we get the first part of the thesis.

\end{proof}


Using Lemma \ref{lm:step1} and Lemma \ref{lm:integral} to estimate $B$, 
we obtain
\begin{align*}
|B|&\leq 
 \frac{{C}^k}{k^{\frac{k}{2}}}\frac{\Gamma(k)}{\Gamma(\frac{k}{2})} 
 (\beta v(0))^{\frac{k}{2}}
\end{align*}
for a suitable constant $C$ which does not depend on $k$.
This estimate is also uniform in $x$. 
Recalling that 
\begin{align*}
e^{\Gamma} e^{-\langle A \rangle_b}e^{W_{A,0}^\xi}=
\sum_{k\geq 0} \frac{(-1)^k\xi^k}{k!}
\prod_{j=1}^k
&\left(
\int \di^3 x_j  g(x_j)  
\sum_{{n}_j\geq 1}
 \int \di {\tilde{W}}_{n_j,x_j}(w_j)
e^{\beta \tilde{\tilde{\mu}} n_j}
\right) \cdot B
\end{align*} we get the estimate 
\begin{align*}
|e^{\Gamma} e^{-\langle A \rangle_b}e^{W_{A,0}^\xi}|
&\leq 
\sum_{k\geq 0} \frac{\xi^k}{k!}
\prod_{j=1}^k
\left(
\int \di^3 x_j  g(x_j)  
\sum_{{n}_j\geq 1}
 \int \di {\tilde{W}}_{n_j,x_j}(w_j)
e^{\beta \tilde{\tilde{\mu}} n_j}
\right) \cdot 
 \frac{{C}^k}{k^{\frac{k}{2}}} 
\frac{\Gamma(k)}{\Gamma(\frac{k}{2})} (\beta v(0))^{\frac{k}{2}}
\\
&\leq 
1+
\sum_{k\geq 1} \frac{\xi^k}{k!}
\left(
\|g\|_1  
\frac{\tilde{C}(\beta \tilde{\tilde{\mu}})}{\beta}
\right)^k \cdot 
 \frac{{C}^k}{k^{\frac{k}{2}}}\frac{\Gamma(k)}{\Gamma(\frac{k}{2})} 
(v(0))^{\frac{k}{2}}
\end{align*}
The asymptotic expansion of $k^{\frac{k}{2}}$ can be obtained by means of the Stirling approximation as  
\[
k^{k} \sim  \frac{e^{k}}{\sqrt{2\pi k}} 
\Gamma\left(k+1\right)
\]
We in particular have the following non optimal estimate
\[
\frac{\Gamma(k)}{k^{\frac{k}{2}}\Gamma(\frac{k}{2})} \leq \left(\frac{2}{e}\right)^{\frac{k}{2}}
\]
The previous sum is thus absolutely convergent and furnish an estimate uniform in $b$.
Which has the form for a suitably chosen $\tilde{{C}}$ which rescales $C$.
\begin{align*}
|e^{\Gamma}e^{-\langle A \rangle_b}e^{W_{A,0}^\xi}|
&\leq 
1+
\sum_{k\geq 1} \frac{\xi^k}{k!}
\left(
\|g\|_1  
\frac{\tilde{C}(\beta \tilde{\tilde{\mu}})}{\beta}
\right)^k \cdot 
\tilde{{C}}^k 
(\sqrt{v(0)})^{k}
\\
& 
\leq
\exp\left( \frac{\tilde{{C}} {v(0)}^{\frac{1}{2}} \|g\|_1}{\beta} \xi
\right)
\end{align*}
for a suitable constant $\tilde{C}$.
The obtained estimate proves that the series in $\xi$ defining $e^{\Gamma} e^{\langle A \rangle_b} e^{W_{A,0}^\xi}$ converges absolutely.
In the same way, we might obtain also a lower bound for $|e^{\Gamma}e^{{W_{A,0}^\xi}}|$ when $\beta$ is sufficiently large and $\xi\leq 1$.
In particular,
\begin{equation}\label{eq:denominator-lower-bound}
\begin{aligned}
|e^{\Gamma}e^{{W_{A,0}^\xi}}| 
&\geq 
1-
\sum_{k\geq 1} \frac{\xi^k}{k!}
\left(
\|g\|_1  
\frac{\tilde{C}(\beta \tilde{\tilde{\mu}})}{\beta}
\right)^k \cdot 
\tilde{C}^k 
(\sqrt{v(0)})^{k}
\\
&\geq 2 -
\exp\left( \frac{\tilde{{C}} {v(0)}^{\frac{1}{2}} \|g\|_1}{\beta} \xi
\right)
\end{aligned}
\end{equation}
hence, 
the the series in powers of the interacting thermal propagators defining the ratio $e^{\Gamma} e^{\langle A \rangle_b} e^{W_{A,0}^\xi}/Z$  satisfies the following bound
\[
\frac{|e^{\Gamma} e^{\langle A\rangle_b}e^{{W_{A,0}^\xi}}|}{|e^{\Gamma}e^{{W_{A,0}^\xi}}|} \leq 
\frac{
\exp\left( \frac{\tilde{{C}} {v(0)}^{\frac{1}{2}} \|g\|_1}{\beta}\xi
\right)}{2 -
\exp\left( \frac{\tilde{{C}} {v(0)}^{\frac{1}{2}} \|g\|_1}{\beta}\xi
\right)}
\]

Combining the observations of this subsubsection we have the proof of Proposition \ref{prop:bounds-no-condenstate}.
\end{proof}





\subsubsection{Bounds in the case of the condensate $T_{1}^\xi \neq 0$}

We recall from Proposition \ref{pr:rel-entropy-condensate} that for $\xi=1$
\[
T_{1} =  \int_0^\beta \di u_1
\int_{0}^\beta  \di u
\langle  g{A}(u_1) \phi_0 ,
\Delta^{\beta,gA}(u_1,u)
gA(u)\phi_0\rangle
\]
Furthermore, using the expression $\mathcal{S}_{\phi_0}-\mathcal{S}_0$ given in \eqref{eq:rel-entropy-condensate}
to rewrite the latter
and introducing the auxiliary parameter $\xi$ as in \eqref{eq:WAxi} we obtain
\[
T_{1}^\xi
=
 \langle  {A}g \phi_0, g \phi_0\rangle_u
+
\xi \langle  K \chi_1 \phi_0, 
(\Delta^{\beta gA}-\Delta^\beta )K  \chi_2 \phi_0 \rangle_u .
\]
where $\chi_i$ are equal to $1$ on the support of $g$ and they are positive smooth and of compact spatial support and they are constant on $u$.
Expanding the thermal propagators in terms of the the path integral representation given in \eqref{eq:path-integral-thermal}
we obtain
\begin{align*}
T_{1}^\xi
&=  \int \di^3 x \tilde{A}g^2 \phi^2_0 
+\xi \int_0^\beta \di u_1
\int_{0}^\beta  \di u \;
\theta(u-u_1) D(u,u_1,0)
+\xi \int_0^\beta \di u_1
\int_{0}^\beta  \di u \sum_{n\geq 1} D(u,u_1,n)
\end{align*}
here $\tilde{A}(x) = \int_0^\beta \di u A(x,u) $ and
\begin{align*}
D(u,&u_1,n):=\\
&\int \di^3 x \int \di^3 y 
K \chi_1(x)\phi_0(x) e^{(\beta n+(u_1-u)) \tilde{\tilde{\mu}}}  \int \di\mu^{x, n\beta + u_1}_{y,u}(\omega) \left(e^{-\int_{u}^{n\beta +u_1} gA(\omega(s),s) \di s } -1\right)
K \chi_2(y)\phi_0(y).
\end{align*}
Dividing the domain of $u$ integration in two parts to resolve the Heaviside step function, 
in 
$0<u < u_1 <\beta$ and 
$0< u_1 < u <\beta$.
Changing the variable of $u$ integration in the two domains respectively to $\delta u = u_1-u $ and to $\delta u = \beta + u_1-u $ and combining the integrals, we obtain 
\begin{align*}
T_{1}^\xi
&= \int \di^3 x \tilde{A}g^2 \phi^2_0 +
+\xi\int_0^\beta \di \delta u
\int_{0}^{\beta}  \di u \sum_{n\geq 0} D(u, u+\delta u,n).
\end{align*}
Using periodicity of $A$ in $u$ and then
integrating by parts in $\di u$, up to a reminder $R$ which can eventually be discarded for an argument similar to the one used in  Lemma 
\ref{lm:lemma-analytic-s} 
we get 
\begin{align*}
T^\xi_{1}
&=  \int \di^3 x \tilde{A}g^2 \phi^2_0 + \xi \beta \int_0^\beta \di \delta u  \sum_{n\geq 0} D_s(\delta u, n) +R
\end{align*}where now
\[
D_s(\delta u, n)\! :=\! 
\int\! \di^3 x \int\! \di^3 y 
K \chi_1(x)\phi_0(x) e^{(\beta n+\delta u) \tilde{\tilde{\mu}}}  \int \di\mu^{x, n\beta + \delta u}_{y,0}(\omega) \left(e^{-\int_{0}^{n\beta +\delta u}\!\! gA(\omega(s),s+\beta) \di s } -1\right)\!
K \chi_2(y)\phi_0(y)
\]

\color{black}

\begin{proposition}\label{lm:estimate-S-condenstante}
It holds that the expansion in powers of the auxiliary parameter $\xi$  of 
$e^{\Gamma} e^{-\langle A \rangle_b}e^{{T}_{1}^\xi}$ gives a series which is absolutely convergent. Actually the following bound holds
 \begin{align*}
|e^{\Gamma} e^{-\langle A \rangle_b}e^{{T}_{1}^\xi}
|
&\leq
\exp \left(  2\beta \phi_0^2 \|K \chi_1\|_2  \| \chi_2 \|_2 \xi\right)
\\
&\leq
\exp (  2\beta \phi_0^2 \epsilon \xi)
\end{align*}
and
\begin{align*}
|e^{\Gamma} e^{T_{1}^\xi}|
&\geq 1+e^{-\frac{\beta }{2}\tilde{v}(g,g)\phi_0^2} -
\exp (  2\beta \phi_0^2 \epsilon \xi)
\end{align*}
\end{proposition}
\begin{proof}
Consider $T_{1}^\xi$.
The two cutoffs $\chi_i$ contained there can be combined to give 
\[
\psi_i = \phi_0 K \chi_i.
\]
In the evaluation of
\[
e^{\Gamma} e^{\langle A\rangle_b} e^{T_{1}^\xi}
\]
since 
\[
e^{\Gamma}e^{-\langle A\rangle_b}
e^{-\langle A\rangle_f}
= e^{-\frac{1}{2}\tilde{v}(b+f,b+f)}\leq 1
\]
we can estimate 
\[
\left(e^{-\int_{0}^{n\beta +\delta u} gA(\omega(s),s+\beta) \di s } -1\right)
\]
by $2$.
The path integral can now be taken, integrating over $\delta u$ and summing over $n$ we obtain
that 
\[
\begin{aligned}
\int \di u \sum_n e^{(\beta n + u)\tilde{\tilde{\mu}}}\int \di^3 x \int \di^3 y \; |\psi_1(x)| \int \di\mu_{y,0}^{x,\beta n+u}(\omega)  |\psi_2(y)|
&=\int \di u \int \di^3 p \;  \overline{\hat{|\psi_1|}(p)}
\hat{|\psi_2|}(p)
\frac{e^{-u (\frac{p^2}{2m}-\tilde{\tilde{\mu}} )}}{1-e^{-\beta (\frac{p^2}{2m}-\tilde{\tilde{\mu}} )}}
\\
&= \int \di^3 p \; 
\frac{1}{\frac{p^2}{2m}-\tilde{\tilde{\mu}}}
\;
\overline{\hat{|\psi_1|}(p)}
\hat{|\psi_2|}(p)
\end{aligned}
\]
Combining we get the estimate
\begin{align*}
|e^{\Gamma} e^{-\langle A \rangle_b}e^{{T}_{1}^\xi}
|
&\leq
\exp \left( 2 \beta \xi
\langle |\psi_1|, K^{-1} |\psi_2| \rangle
\right)
\\
&\leq
\exp \left( 2 \beta \xi
\langle |K\chi_1 \phi_0|, K^{-1} |K\chi_2 \phi_0|| \rangle 
\right)
\end{align*}
Notice that, in the limit where $\chi_1$ and $\chi_2$ tends to $1$, 
($\chi_i^n(x) = \chi_i(x/n)$ for large $n$ and for $\epsilon = \tilde{\tilde{\mu}} \to 0$)
we have that $\|\psi_i\|_2 \to 0 $  as $\epsilon n^{3/2} + 1/\sqrt{n}$.
Furthermore,
\[
\|\frac{1}{p^2/2m-\tilde{\tilde{\mu}}}\|_\infty\leq \frac{1}{\epsilon}
\]
Combining the estimates we have by Cauchy Schwartz inequality
\[
2 \beta
\langle |K\chi_1 \phi_0|, K^{-1} |K\chi_2 \phi_0|| \rangle
\leq 2\beta \phi_0^2 \|K \chi_1\|_2  \| \chi_2 \|_2
\]
in the limit where $n$ in $\chi_1$ diverges and $\epsilon$ to $0$ in the appropriate way, we have that the argument of the exponential tends to $0$, independently on $\beta$ and $
\phi^2_0$.

To get the lower bound we observe that for $\xi=0$ 
\[
\left. e^{\Gamma} e^{T_{1}^\xi} \right|_{\xi=0} = 
e^{\Gamma}e^{\int du A(u) g^2\phi_0^2} 
=e^{-\frac{\beta }{2}\tilde{v}(g,g)\phi_0^2} 
\]
hence, arguing as in \eqref{eq:denominator-lower-bound} we get the thesis.
\end{proof}

\subsubsection{Combined estimates}


Combining the results of the previous subsections 
we have all the elements to evaluate the expectation values of the two-point functions of the interacting theory.
We actually have that for any $f,h\in\mathcal{H}$,
\[
e^{\Gamma}(\langle f, \Omega^{\phi_0}h\rangle e^{W_A})
\]
can be expanded as a series in powers of the auxiliary parameter $\xi$ given in \eqref{eq:WAxi}. 
To prove absolute convergence of that power series, we combine the estimates of Proposition \ref{prop:bounds-no-condenstate} to analyze the contribution  of $\mathcal{S}_0 -\omega^\beta(V_A)$ with Proposition \ref{lm:estimate-S-condenstante} to estimate the contribution of $\mathcal{S}_{\phi_0}- \mathcal{S}_0$. 
We thus have that 
with a suitable choice of $\beta$, of $g$, of $\epsilon$, of $\phi_0$ and if $v$ is sufficiently small the correlation functions of the equilibrium state for the interacting theories can be given and are finite. We have actually the following theorem
\begin{theorem}\label{th:convergence}
    Combining the estimates we have that
    \[
    I=\omega^{\beta V} (\Psi(f)\Psi^*(h)) = \left. \frac{e^{\Gamma} \langle f,\Omega^{\phi_0}h\rangle Z_A }{e^{\Gamma} Z_A} \right|_{A=0}
    \]
    and it holds that the truncated series in powers of the auxiliary parameter $\xi$ 
    which results after the application of $e^{\Gamma}$ 
    are bounded by
    \[
    |I| \leq \frac{C_{f,g}(1 + \phi_0^2)e^{R}}{1+e^{- \frac{\beta\phi_0^4}{2} \tilde{v}(g,g)}-e^{R}}  C
    \]
    for $\xi\leq 1$, $C_{f,g}$ a positive constant depending on $f,g$, and uniformly on the order of truncation
where 
\[
R  = \epsilon \beta  \phi_0^{2} \|g\| 
+ \frac{\sqrt{v(0)}}{{\beta}} \|g\|_1 \tilde{{C}}
\]
and $C$ is a suitable bound of $\Delta^{\beta}(0)$.

    For $\xi\leq 1$, the series is thus absolutely convergent if the parameter of the theory are chosen in such a way that 
    \[
    1+e^{- \frac{\beta\phi_0^4}{2} \tilde{v}(g,g)} > e^{R}
    \]
\end{theorem}
\begin{proof}
Recall that 
\[
\langle f,\Omega^{\phi_0} h\rangle = \langle f,\Omega h\rangle + \omega^{\beta A}(\Psi(f))\omega^{\beta A}(\Psi^*(h)).
\]
Furthermore, 
\[
\langle f,\Omega h\rangle = 
\int \di^3 x\di^3 y  f(x)\sum_n \int \di\mu_{x,0}^{y,n \beta}  e^{-\sum_{l=0}^{n-1}
\int_0^{\beta } \di s  
\di u   A(\omega(\beta l +s),s) } h(y)
\]
and
\[
\sum_{n} \int \di\mu_{x,0}^{y,n \beta}  e^{-\sum_{l=0}^{n-1}
\int_0^{\beta } \di s  
\di u   A(\omega(\beta l +s),s) }
= \sum_n \int \di\mu_{x,0}^{y,n \beta}  e^{-\langle A\rangle_{h_\omega}}
\]
for a suitable $h_{\omega}$.
Furthermore, the sum over $n$ is bounded as in Lemma \ref{lm:norm-state}.

Using a method similar 
to the one discussed in Proposition \ref{lm:estimate-S-condenstante} we also have that 
\begin{align*}
\omega^{\beta A}(\Psi(f)) 
&= \int_0^\beta \di u \langle f, \Delta^{\beta A}(0,u) A(u) g\phi_0\rangle  
\\
&= 
\int_0^\beta \di u \langle f, \Delta^{\beta A}(0,u) A(u) g\chi \phi_0\rangle  
\\
&= \langle f, \chi \phi_0\rangle
-
\int_0^\beta \di u \langle f, \Delta^{\beta A}(0,u) (K+\frac{\partial}{\partial u})\chi \phi_0\rangle
\end{align*}
for a suitable smooth compactly supported function $\chi$ which is $1$ on the support of $g$.
Hence also $\omega^{\beta A}(\Psi(f))$ admits an expansion in term of an exponential of a suitably smeared $A$
\[
\omega^{\beta A}(\Psi(f)) =
\langle f, \chi \phi_0\rangle
-
\sum_n \int \di^3 x \di^3 y f(x)\int \di\mu_{x,0}^{y,\beta n} e^{-\langle A\rangle _{h_\omega}} (K+\frac{\partial}{\partial u})\chi(y)\phi_0
\]
and a similar expansion can be obtained for $\omega^{\beta A}(\Psi^*(h))$.
The proof of this theorem can now be obtained combining the estimates given in Proposition \ref{prop:bounds-no-condenstate} and in Proposition \ref{lm:estimate-S-condenstante}.
\end{proof}




\appendix

%
%
%

\section*{Acknowledgments}
The authors want to thank Benjamin Schlein for illuminating discussions on BEC and Klaus Fredenhagen for useful comments about the structure of interacting quantum field theory in the non relativistic limit. S.G. further and gratefully acknowledges Benjamin Schlein for his financial support and warm hospitality at the Department of Mathematics, University of Zürich. Additional support was provided through a Short-Term Scientific Mission funded by COST Action CA21109 – CaLISTA, supported by COST (European Cooperation in Science and Technology). 
The authors are  grateful for the support of the National Group of Mathematical Physics (GNFM-INdAM). 
The research presented in this paper was supported in part by the MIUR Excellence Department Project 2023-2027 awarded to the Department of Mathematics of the University of Genova, CUPD33C23001110001.
\\

{\bf Data availability:} 
Data sharing not applicable to this article as no datasets were generated or analyzed during
the current study. 
\\

{\bf Conflicts of interest statement:} 
All authors have no conflicts of interest to declare.

\section{Counting the number of contributions} \label{ap:counting}
To count the number of diagrams we can use ideas similar to those present in \cite{Cvitanovic}
Use a one dimensional theory. This amount to consider fields as numbers and the propagator is then $1$. 
With this idea,
for an uncharged field, the number of diagrams (also the non connected ones) in 
\[
\mathcal{T}(A^{2N})
\]
is computed to be 
\[
\left. 
e^{\frac{1}{2} \frac{\partial^2}{\partial A^2}}A^{2N} 
\right|_{A=0} = \frac{(2N)!}{2^N N!} = (2N-1)!!
\]
and it is $0$ for odd powers of the field.
For a charged field
\[
\left. 
e^{ \frac{\partial^2}{\partial \Phi\partial\Phi^*}}|\Phi|^{2N} 
\right|_{\Phi=\Phi^*=0} = N! 
\]
and $0$ for an uneven number of fields $\Phi$ and $\Phi^{*}$. 
The restriction to connected diagrams does not alter the leading asymptotic behavior.
It contributes at most with a factor.



\printbibliography
\end{document}